\documentclass[wcp]{jmlr}

\usepackage{longtable}

\usepackage{enumerate}
\usepackage{enumitem}
\usepackage{wrapfig}
\usepackage{multicol}
\usepackage{mathtools}

\allowdisplaybreaks

\usepackage{booktabs}

\makeatletter
\let\Ginclude@graphics\@org@Ginclude@graphics 
\makeatother

\jmlrvolume{189}
\jmlryear{2022}
\jmlrworkshop{ACML 2022}

\title[Noise Robust Core-stable Coalitions of Hedonic Games]{Noise Robust Core-stable Coalitions of Hedonic Games}

  \author{\Name{Prashant Trivedi} \Email{trivedi.prashant15@iitb.ac.in}\\
  \addr IEOR, Indian Institute of Technology Bombay
  \AND
  \Name{Nandyala Hemachandra} \Email{nh@iitb.ac.in}\\
  \addr IEOR, Indian Institute of Technology Bombay
 }

\editors{Emtiyaz Khan and Mehmet G\"{o}nen}

\begin{document}

\maketitle

\begin{abstract}
In this work, we consider the coalition formation games with an additional component, `noisy preferences'. Moreover, such noisy preferences are available only for a sample of coalitions. We propose a multiplicative noise model (equivalent to an additive noise model) and obtain the prediction probability, defined as the probability that the estimated PAC core-stable partition of the \emph{noisy} game is also PAC core-stable for the \emph{unknown noise-free} game. This prediction probability depends on the probability of a combinatorial construct called an `agreement event'. We explicitly obtain the agreement probability for $n$ agent noisy game with $l\geq 2$ support noise distribution. For a user-given satisfaction value on this probability, we identify the noise regimes for which an estimated partition is noise robust; that is, it is PAC core-stable in both noisy and noise-free games. We obtain similar robustness results when the estimated partition is not PAC core-stable. These noise regimes correspond to the level sets of the agreement probability function and are non-convex sets. Moreover, an important fact is that the prediction probability can be high even if high noise values occur with a high probability. Further, for a class of top-responsive hedonic games, we obtain the bounds on the extra noisy samples required to get noise robustness with a user-given satisfaction value. 

We completely solve the noise robustness problem of a $2$ agent hedonic game. In particular, we obtain the prediction probability function for $l=2$ and $l=3$ noise support cases. For $l=2$, the prediction probability is convex in noise probability, but the noise robust regime is non-convex. Its minimum value, called the safety value, is 0.62; so, below 0.62, the noise robust regime is the entire probability simplex. However, for $l \geq 3$, the prediction probability is non-convex; so, the safety value is the global minima of a non-convex function and is computationally hard.  
\end{abstract}

\section{Introduction}
\label{sec:intro}
Coalition formation games are of great interest to researchers because they model natural interactions among multi-agent societies.
The coalition formation process can be formalized using the framework of hedonic games. In these games, each agent has a preference over the coalitions they form with the other agents. An outcome of a hedonic game consists of dividing the agent set into disjoint coalitions called partition. Such a partition is referred to as \textit{coalition structure}. A desirable property in hedonic games is the formation of a stable coalition structure. However, any stability notion \citep{bogomolnaia2002stability,Aziz_savani} assumes the complete information of each agent's preferences, i.e., the entire ranking of coalitions by each agent is known. This is one of the strong assumptions in hedonic games. Nonetheless, there is significant work in finding a stable partition of the agent set, if it exists \citep{brandt2016handbook}.

Authors in \cite{sliwinski2017learning} relax the assumption of complete information and assume that the preferences over only some coalitions are available; they introduce the notion of $\epsilon$-Probably Approximately Correct ($\epsilon$-PAC) stability to learn the stable outcome of the hedonic game. Apart from the assumption about the complete information, we can have preferences corrupted by noise, i.e., the exact preferences of agents are not available; instead, the preferences with errors are observed. We call such observed erroneous preferences, \textit{noisy preferences}. A consequence of these noisy preferences is that a partition that is not stable in a noisy game can be stable in a noise-free game with non-trivial probability or vice-versa.

In this work, one of our goals is to find the probability that a stable partition learned from the observed noisy sample is the same as that of the stable partition of the \textit{unknown noise-free} game (Sec. \ref{sec: n_agent_partial_info}). We obtain similar results when one starts with a particular partition that is not PAC stable for the noisy game. In such a case, we are interested in the probability that the estimated partition is also not core-stable for a noise-free game (Sec. \ref{sec: Non-existence_parital_info}). We call these probabilities the \textit{prediction probabilities}.
These prediction probabilities depend on a probability of an event called the `agreement event'. We also obtain the noise regimes where the agreement probabilities are more than a user-given threshold.

As a motivation, let us consider a stylized model of a market for a specific product that three manufacturers $N = \{1, 2, 3\}$ serve. Each manufacturer has preferences, denoted by $\succ_i, ~\forall~i\in N$, over the coalitions they want to form with other manufacturers. Based on their preferences, a market analyst would like to predict the coalition structure that these three manufacturers form. However, these preferences being private to manufacturers, the market analyst collects them through a noisy channel (or estimates them based on the market's history). For simplicity, assume that the analyst has noisy preferences, denoted by $\succ^{\prime}_i,~ \forall~ i\in N$, of all the agents (a complete information model) as in the game \eqref{eqn: noisy_game_motivation} below:
\begin{minipage}{0.47\textwidth}
\begin{equation}
\label{eqn: noisy_game_motivation}
\begin{aligned} 
\textcolor{red}{\{12\} \succ_1^{\prime} \{1\}} \succ_1^{\prime} \{123\} \succ_1^{\prime} \{13\}
\\ 
\{12\} \succ_2^{\prime} \{2\} \succ_2^{\prime} \{123\} \succ_2^{\prime} \{23\}
\\ 
\{123\} \succ_3^{\prime} \{23\} \succ_3^{\prime} \{13\} \succ_3^{\prime} \{3\}
\end{aligned}
\end{equation}    
\end{minipage}\quad
\begin{minipage}{0.47\textwidth}
\begin{equation}
\label{eqn: noise-free_game_motivation}
\begin{aligned} 
\textcolor{red}{\{1\} \succ_1 \{12\}} \succ_1 \{123\} \succ_1 \{13\}
\\ 
\{12\} \succ_2 \{2\} \succ_2 \{123\} \succ_2 \{23\}
\\ 
\{123\} \succ_3 \{23\} \succ_3 \{13\} \succ_3 \{3\}
\end{aligned}
\end{equation}
\end{minipage}

The noisy core-stable partition as predicted by the market analyst is $\tilde{\pi} = \{ \{12\}, \{3\}\}$. However, suppose the noise-free preferences are as in game \eqref{eqn: noise-free_game_motivation} (these are not known to market analyst). Based on these noise-free preferences the unique core-stable partition is $\pi = \{\{1\}, \{2\}, \{3\}\}$. 
So, while the market analyst concludes that the manufacturers form a coalition based on the available noisy preferences, they will not. Thus, the market analyst needs to know the prediction probability, the probability that the predicted partition based on the available noisy preferences is the same as the partition of the unknown noise-free game in \eqref{eqn: noise-free_game_motivation}.
An interesting phenomenon in the noisy hedonic game is that even if the market analyst misses identifying a core-stable partition, the market has one with non-trivial probability. We consider this in Sec. \ref{sec: Non-existence_parital_info}, via their complimentary event. 

As a generalization to the above three manufacturers' model, we assume that a learner has preferences over some coalitions collected via a noisy channel. Based on these noisy preferences, the learner's task is to predict the core-stable partition for an unknown noise-free game based on these partial and noisy preferences. We propose a noise model to investigate noise regimes where the predicted noisy partition is the same as an unknown noise-free partition. It addresses an important aspect of noise-robustness, meaning these partitions are the same with high probability. Specifically, our major contributions are:
\\
\noindent \textbf{(a)} In Sec. \ref{sec: n_agent_partial_info}, we propose a multiplicative noise model and obtain the prediction probability.  This probability depends on a combinatorial construct called `agreement event'. For a user-given value on agreement probability, we obtain the noise regimes where an estimated partition of the noisy game is noise-robust.
\\
\noindent \textbf{(b)} In Subsec.  \ref{subsec: sample_complexity}, we obtain the lower and upper bounds on the number of noisy samples required to get PAC stable partition for the top-responsive class of hedonic games. 
\\
\noindent \textbf{(c)} In Sec. \ref{sec: Non-existence_parital_info}, we obtain the prediction probability function that a partition $\tilde{\pi}$ is not PAC stable for the noise-free game given that it is not PAC stable for the noisy game. 
\\
\noindent \textbf{(d)} In Sec. \ref{sec: 2 agent_full_info}, we consider a noisy game with 2 agents with complete information on each agent's preferences. The allowable noise regimes for noise-robustness are non-convex, even though the agreement probability is a convex function.
\\
\noindent \textbf{(e)} We now mention some observations of 2 agent game. The prediction probability is non-convex when the noise distribution has $l \ (\geq 3)$ support. Thus, computing the safety value, i.e., the minimum prediction probability, is computationally hard. So, for user satisfaction values below this safety value, the prediction probability is 1, regardless of the noise values and their probabilities.
Also, the noise values that render a user given minimum prediction probability form non-contiguous regions (superlevel sets). A counter-intuitive fact is that the prediction probabilities can be high for some high noise values occurring with high probability. A simple illustration is in the 2 support case, where the prediction probability is 1 even when the value of both agents is inflated with probability 1.

\subsection{Notations and preliminaries}
\label{sec: notations_prelims}
This Sec. provides some notations, definitions, and other related backgrounds that we use in the paper subsequently. Let $N = \{1,2,\dots, n\}$ be the set of agents and for each agent $i\in N$, let $\mathcal{C}_i  = \{S\subseteq N ~ |~ i\in S\}$ be the set of coalitions containing agent $i$. A hedonic game is a pair $(N, \succeq)$, where $\succeq = (\succeq_1, \succeq_2, \dots, \succeq_n)$. Here $\succeq_i$ is a reflexive, transitive, and complete preference ranking of agent $i\in N$ over the set $\mathcal{C}_i$. The preference $\succeq_i$ of agent $i \in N$ represents its willingness to form a coalition with other agents.

For any two distinct coalitions $S, T\subseteq \mathcal{C}_i$ we say $S \succ_i T$ if agent $i\in N$ prefers coalition $S$ over $T$. Also, $S\sim_i T$ iff $S\succeq_i T$ and $T\succeq_i S$, that is agent $i$ is indifferent to coalition $S$ and $T$. Since the preferences are reflexive, transitive and complete there exists a value function $\textbf{\textit{v}}: S\subseteq N \mapsto \mathbb{R}^{|S|}$ such that $\textbf{\textit{v}}(S) = (v_i(S))_{i\in S}$ \footnote{Note that $(v_i(S))_{i\in S}$ is a vector of size $|S|$ with each element $v_i(S)$ for agent $i\in S$.}, where $v_i(S) \in \mathbb{R}^{+}$ is the valuation of an agent $i$ in coalition $S$. For any coalitions $S,T \in \mathcal{C}_i$, it satisfies that $S\succeq_i T \iff {v}_i(S) \geq {v}_i(T)$ \citep{mas1995microeconomic, narahari2014game}. The valuation $v_i(S)$ often depends on value $v_i(j) \in \mathbb{R}^{+}$ of agent $j$ in the eyes of agent $i$, here $i, j\in S$. We use $(N, \textbf{\textit{v}})$ to denote the hedonic game.

A typical partition of the  agent set in the hedonic game $(N, \textbf{\textit{v}})$ is denoted by $\pi$. Let the coalition containing $i\in N$ in partition $\pi$ be $\pi(i)$. The hedonic game's outcome is finding a `stable' partition according to some stability criterion. A partition is `stable' if no agent or a group of agents can deviate from it to reach a subjectively better outcome. Various stability criteria are introduced in \cite{bogomolnaia2002stability} and are nicely reviewed by \cite{Aziz_savani}. However, in this paper, we use core, one of the popular stability criteria. A coalition $S$ \textit{core blocks} a partition $\pi$, if every agent $i$ in coalition $S$ strictly prefers $S$ to $\pi(i)$, i.e., $S \succ_i \pi(i),~\forall ~ i \in S$. Further,  a coalition structure $\pi$ is said to be \textit{core-stable} if there is no coalition that core blocks $\pi$, meaning there is at least one agent $i\in S$ who prefers $\pi(i)$ over $S$, i.e., $\pi(i) \succeq_i S$.

Recently, for a partial information hedonic game, authors in \cite{sliwinski2017learning} have proposed the PAC learning framework to find a $\epsilon$-PAC stable outcome for several classes of hedonic games. We briefly describe the $\epsilon$-PAC stability framework here \citep{sliwinski2017learning}. Given a sample $\mathcal{S} =  \{(S_1,\textbf{\textit{v}}(S_1)),\dots, (S_{m}, \textbf{\textit{v}}(S_m))\}$, where $S_1, S_2, \dots, S_m$ are drawn \textit{i.i.d.} from a distribution over $2^N$ and the corresponding values $\textbf{\textit{v}}$'s are obtained from $\mathcal{D}$. An algorithm $\mathcal{A}$ is said to PAC stabilize a class $\mathcal{H}$ of hedonic games if for any hedonic game $(N, \textbf{\textit{v}}) \in \mathcal{H}$, after seeing examples in $\mathcal{S}$ 
it can propose a partition $\pi$ that is unlikely to be core blocked by a coalition sampled from $\mathcal{D}$ with high probability. Formally, for any error and the confidence parameter $\epsilon, \delta>0$, a partition $\pi$ is $\epsilon$-PAC stable under $\mathcal{D}$ if $\mathcal{A}$ outputs a $\epsilon$-PAC stable coalition structure or reports that the core is empty, i.e., 
\begin{equation}
\label{eqn: PAC_stability}
\mathbb{P}_{\mathcal{S}}[\mathbb{P}_{T\sim \mathcal{D}} [T~ \text{\textit{core blocks}}~ \pi \text{\textit{ in noise-free game }} (N, \textbf{\textit{v}})] < \epsilon ] \geq 1-\delta,
\end{equation}
here, the number of samples $m$ are required to be polynomial in $n, \frac{1}{\epsilon}~ \text{and} \log\frac{1}{\delta}$.

As mentioned above, the $\epsilon$-PAC stability notion assumes the correct preferences over the sample of coalitions. However, it is not the case in most realistic scenarios. Often the preferences are erroneous, i.e., corrupted by noise. In this work, we relax both the assumptions of correct and complete knowledge of the preferences. Let the value of each agent in any coalition be corrupted by an unknown noise distribution, $\mathcal{N}$. We denote the complete, reflexive and transitive noisy preferences by $\succeq^{\prime} = (\succeq_1^{\prime}, \succeq_2^{\prime}, \dots, \succeq_n^{\prime})$. The noisy hedonic game is therefore represented by $(N, \succeq^{\prime})$ or equivalently $(N, \tilde{\textbf{\textit{v}}})$, where $\tilde{\textbf{\textit{v}}}(S) = (\tilde{v}_i(S))_{i\in S}$ is such that $\tilde{v}_i(S) \in \mathbb{R}^+$. Formally, we are given a sample $\mathcal{\tilde{S}} = \{(S_1,\tilde{\textbf{\textit{v}}}(S_1)), \dots, (S_{\tilde{m}}, \tilde{\textbf{\textit{v}}}(S_{\tilde{m}}))\}$ from the noisy hedonic game $(N, \tilde{\textbf{\textit{v}}})$. Here $S_1, S_2, \dots, S_m$ are drawn \textit{i.i.d.} from a distribution over $2^N$ and the corresponding values $\tilde{\textbf{\textit{v}}}$'s are obtained from $\mathcal{\widetilde{D}}$.
So, we can find a $\tilde{\epsilon}$-PAC stable partition (this can be done by using an algorithm similar to one given in say, \cite{sliwinski2017learning,alcalde2004researching}) if it exists. 
Let $\tilde{\pi}$ be an $\tilde{\epsilon}$-PAC stable partition of the noisy hedonic game, i.e., with probability at least $1-\delta$, we have,
\begin{equation}
\label{eqn: PAC_stability_noisy}
\mathbb{P}_{\tilde{\mathcal{S}}} [ \mathbb{P}_{T\sim \mathcal{\widetilde{D}}} [T~ \text{\textit{core blocks}}~ \tilde{\pi} \text{\textit{ in noisy game }} (N, \tilde{\textbf{\textit{v}}})] < \tilde{\epsilon}] \geq 1-\delta.
\end{equation} 
Again, the number of samples required are $\tilde{m}$ which is polynomial in $n, \frac{1}{\tilde{\epsilon}}$, and $\log \frac{1}{\delta}$. The entire paper uses the inner probability given in Equation \eqref{eqn: PAC_stability} for the noisy game. However, for the noise-free game, we are interested in the probability given in Equation \eqref{eqn: predi_prob_with_core} below. This is because we only have samples from the noisy game; hence, the outer probability is taken on noisy samples for noisy and noise-free games.

Let $\alpha_i(S) \sim \mathcal{N}$ be the noise realized to an agent $i \in S\subseteq N$. We assume that for each agent $i\in S$, the noise is the same, i.e., $\alpha_i(S) = \alpha(S) \in \mathbb{R}^{+},~\forall~i\in S$. To ensure noise distribution support, $\mathcal{N}_{sp}$ is non-empty we assume that it contains $1$ and other noise values. So, the noisy value is $\tilde{v}_i(S) \coloneqq \alpha(S) \cdot v_i(S), ~ \forall~ i\in S \subseteq N$. We call this noise model the \textit{multiplicative noise} model; this is equivalent to the additive noise model as given in Remark \ref{remark: additive} below.
In Sections \ref{subsec: n_agent_l_support_partial_info}, and \ref{subsec: n_agents_l_support_non_existence}, we also consider scaling at various levels by taking $l\geq 2$ support on the noise distribution (the same level for all members of a given coalition). Another motivation for the same noise level scaling for each agent in the coalition is the following: 
A learner is collecting the valuation of each coalition via a noisy channel. So, we assume that a noisy channel affects the value of the entire coalition by the same amount. Hence, each agent in a coalition will have the same noise impact, irrespective of its identity. However, suppose an agent $i$ is a member of two coalitions $S, T \in \mathcal{C}_i$. The noise valuation of agent $i$ in coalition $S$ is $\alpha(S) v_i(S)$, and $\alpha(T)v_i(T)$ in coalition $T$, so, we also have different noise values for the same agent depending on the coalition.  Moreover, the assumption of common scaling $\alpha(S)$ is necessary to carry out the Probably Approximately Correct (PAC) analysis. The PAC stability definition uses a hypothesis class, in our case,  the class of hedonic games. The common $\alpha(S)$ preserves the class of hedonic games under noise, which need not be the case when we scale the value of each coalition at the individual agent level. For example, if the noise-free game belongs to the class of additively separable hedonic games (ASHGs), the noisy game with agent dependent noise scaling $\alpha_i(S), ~\forall ~i\in S$ may not be ASHG, but it is within ASHG class with common noise scaling $\alpha(S)$. So, we use a common scaling that restricts noisy and noise-free games to the same class. We believe this assumption can be relaxed by taking the larger class of hedonic games; however, we might need additional conditions to ensure the class-preserving property.

It is important to note that our noisy hedonic game setup can be reduced to the noise-free setup in a very specialized setting, i.e., only if $\alpha(S) =1$ for all coalitions $S$.

Suppose $\tilde{\pi}$ is any partition of the noisy game $(N, \tilde{\textbf{\textit{v}}})$. We aim to find the probability that any $T\sim \mathcal{\widetilde{D}}$ core blocks $\tilde{\pi}$ in the noise-free game $(N, \textbf{\textit{v}})$, i.e., 
\begin{equation}
\label{eqn: predi_prob_with_core}
    \mathbb{P}_{T\sim \tilde{\mathcal{D}}} [T~ \text{\textit{core blocks}}~ \tilde{\pi} \text{\textit{ in noise-free game }} (N, \textbf{\textit{v}})].
\end{equation}
We call the above probability, \textit{prediction probability}. In each case, i.e., when $\tilde{\pi}$ is $\tilde{\epsilon}$-PAC stable partition of the noisy game or not, we bound these prediction probability in Sections \ref{sec: n_agent_partial_info} and \ref{sec: Non-existence_parital_info}, respectively. Prediction probability is a performance measure associated with noise robustness, as defined below:
\begin{definition}[$\zeta$ noise-robust core-stable partition $\tilde{\pi}$]
\label{def: zeta_noise_robustness}
A partition $\tilde{\pi}$ is $\zeta$ noise-robust core-stable partition if (a) $\tilde{\pi}$ is $\tilde{\epsilon}$-PAC stable partition of noisy game $(N, \tilde{\textbf{\textit{v}}})$, and (b) prediction probability in Equation \eqref{eqn: predi_prob_with_core} is less than $\epsilon$, where $\epsilon = 1 - (1 - \tilde{\epsilon})\zeta$ with $\zeta \in (0, 1]$.
\end{definition}
\begin{definition}[$\eta$ noise-robust non core-stable partition $\tilde{\pi}$]
\label{def: eta_noise_robustness}
A partition $\tilde{\pi}$ is $\eta$
noise-robust non core-stable partition if (a) $\tilde{\pi}$ is not $\tilde{\epsilon}$-PAC stable partition of noisy game $(N, \tilde{\textbf{\textit{v}}})$ and (b) prediction  probability in Equation \eqref{eqn: predi_prob_with_core} is more than $1-\epsilon$, where $\epsilon = 1- (1 - \tilde{\epsilon})\eta$ with $\eta \in (0, 1]$.
\end{definition}

\begin{remark}
\label{remark: additive}
\underline{Additive noise model:} Our noise model is a fairly generic one. For example, if the noise is additive, i.e., $\tilde{v}_i(S) = \alpha(S)  + v_i(S)$ then, taking exponential on both sides, we have $e^{\tilde{v}_i(S)} = e^{\alpha(S)  + v_i(S)} = e^{\alpha(S)} \cdot e^{v_i(S)}$. With $\tilde{V}_i(S) = e^{\tilde{v}_i(S)}$, $\Gamma (S) = e^{\alpha(S)}$, and $V_i(S) = e^{v_i(S)}$, we have $\tilde{V}_i(S) = \Gamma(S) V_i(S)$. Hence, for robustness to an additive noise model, one can reduce it to  a noisy hedonic game $(N, \widetilde{\textbf{\textit{V}}})$ with multiplicative noise $\Gamma(S)$.
\end{remark}

\begin{remark}
Note that in Equation \eqref{eqn: predi_prob_with_core} we use noise-free values $\textit{\textbf{v}}$'s to check whether a coalition can potentially block a given noisy core-stable partition $\tilde{\pi}$. However, we only have samples from the noisy game, so we use $T \sim \mathcal{\widetilde{D}}$ instead of $T\sim \mathcal{D}$.
\end{remark}

\section{Partial information noisy game with $\tilde{\pi}$ as $\tilde{\epsilon}$-PAC stable partition}
\label{sec: n_agent_partial_info}
Let $\mathcal{\tilde{S}} = \{(S_1,\tilde{\textbf{\textit{v}}}(S_1),\dots,(S_{\tilde{m}}, \tilde{\textbf{\textit{v}}}(S_{\tilde{m}}))\}$ be a sample of coalitions drawn \textit{i.i.d} from the distribution $\mathcal{\widetilde{D}} = \mathcal{D} \times \mathcal{N}$ over $2^N$. 
Let $\tilde{\pi}$ be $\tilde{\epsilon}$-PAC stable outcome of noisy game $(N, \tilde{\textbf{\textit{v}}})$. Therefore, with probability at least $1-\delta$, $\forall~ \tilde{\epsilon} >0$, we have
\begin{equation}
\begin{aligned}
\mathbb{P}_{T\sim \mathcal{\widetilde{D}}}[T~ core~ blocks~ \tilde{\pi}] < \tilde{\epsilon}, 
~~ or \hspace{3mm} \mathbb{P}_{T\sim \tilde{\mathcal{D}}} [ \tilde{v}_i(T) > \tilde{v}_i(\tilde{\pi}(i)),~~ \forall~i\in T] &< \tilde{\epsilon},
\\
or \hspace{3mm} \mathbb{P}_{T\sim \tilde{\mathcal{D}}} [\cup_{i\in T}~~ \tilde{v}_i(\tilde{\pi}(i)) \geq \tilde{v}_i(T)] &\geq 1- \tilde{\epsilon}.
\end{aligned}
\label{eqn: core_block_noisy}
\end{equation}
For an unknown noise-free hedonic game $(N, \textbf{\textit{v}})$, we now find the prediction probability given in Equation (\ref{eqn: predi_prob_with_core}). To this end, we first define set $\mathcal{R}(T)$ for any coalition $T$ as $\mathcal{R}(T) \coloneqq
\{\tilde{\pi}(i)\in \tilde{\pi}~|~i\in T\}$, i.e., for all agents $i \in T$, it is the set of all coalitions containing agent $i$ in the partition $\tilde{\pi}$. Moreover, for any coalition $T$, and partition $\tilde{\pi}$, we define an agreement event $M(\tilde{\pi}, T)$ containing the set of all noise levels $\alpha(\tilde{\pi}(i))$ and $\alpha(T)$ such that all the coalitions $\tilde{\pi}(i) \in \mathcal{R}(T)$ are preferred over coalition $T$ by every agent $i\in T$ in both noisy and noise-free game. Formally, it is defined as
\begin{small}
\begin{equation*}
\label{eqn: agreement_prob_M_pi}
M(\tilde{\pi}, T) \coloneqq \{(\{\alpha(\tilde{\pi}(i))\}_{\tilde{\pi}(i)\in \mathcal{R}(T)}, \alpha(T)) :  \cap_{i\in T} \{v_i(\tilde{\pi}(i)) \geq v_i(T) ~\cap~ \alpha(\tilde{\pi}(i)){v}_i(\tilde{\pi}(i))  \geq \alpha(T){v}_i(T)\}\}.
\end{equation*}
\end{small}
Let $f_T(\textbf{\textit{p}}, \boldsymbol{\alpha}) \coloneqq \mathbb{P}_{T\sim \mathcal{\widetilde{D}}}[M(\tilde{\pi}, T)]$ be the probability of agreement event $M(\tilde{\pi}, T)$, where $\textbf{\textit{p}}$ is the probability mass function of noise values $\boldsymbol{\alpha}$ \footnote{Here $\boldsymbol{\alpha}$ contains all possible noise values, and $\textbf{\textit{p}}$ is the probability mass function of noise values in $\boldsymbol{\alpha}$.}. Note that $M(\tilde{\pi}, T)$ is not known, since the noise-free values $v_i(T)$ and $v_i(\tilde{\pi}(i))$ are not known. However, for $l\geq 2$ support noise distribution we obtain explicit expressions for $f_T(\textbf{\textit{p}}, \boldsymbol{\alpha})$ in Sec. \ref{subsec: n_agent_l_support_partial_info}. We also use $f_T(\textbf{\textit{p}}, \boldsymbol{\alpha})$ later as user satisfaction value. The following Theorem gives probability that unknown noise-free game $(N, \textbf{\textit{v}})$ has $\tilde{\pi}$ as $\epsilon$-PAC stable partition 
($\epsilon$ is identified in the Theorem \ref{thm: prob_same_partition_n_partial_info} below in terms of $\textbf{\textit{p}}, \boldsymbol{\alpha}$ and $\tilde{\epsilon}$) if noisy game $(N, \tilde{\textbf{\textit{v}}})$ has $\tilde{\pi}$ as $\tilde{\epsilon}$-PAC stable partition. 

\begin{theorem}
\label{thm: prob_same_partition_n_partial_info}
Let $\tilde{\pi}$ be $\tilde{\epsilon}$-PAC stable outcome of the noisy game $(N, \tilde{\textbf{\textit{v}}})$. Then, $\tilde{\pi}$ is $\epsilon$-PAC stable for noise-free game $(N, \textbf{\textit{v}})$, i.e., $\mathbb{P}_{T\sim \tilde{\mathcal{D}}}[\cup_{i\in T}~~ v_i(\tilde{\pi}(i)) \geq v_i(T) ] \geq 1-\epsilon$, where
$\epsilon > 0$ satisfies $(1-\tilde{\epsilon}) f_T(\textbf{p}, \boldsymbol{\alpha}) = 1- \epsilon$ with $f_T(\textbf{p}, \boldsymbol{\alpha}) = \mathbb{P}[M(\tilde{\pi}, T)]$.
\end{theorem}
\begin{proof}
Consider the following probability
\begin{equation*}
\begin{aligned}
\mathbb{P}_{T\sim \tilde{\mathcal{D}}}[\cup_{i\in T}~ v_i(\tilde{\pi}(i)) \geq v_i(T) ] 
&\geq \mathbb{P}_{T\sim \mathcal{\widetilde{D}}}[\cup_{i\in T} v_i(\tilde{\pi}(i)) \geq v_i(T)  | \cup_{j\in T} \tilde{v}_j(\tilde{\pi}(j)) \geq \tilde{v}_j(T)] 
\\
&  ~~~~~\times~ \mathbb{P}_{T\sim \mathcal{\widetilde{D}}} [\cup_{j\in T} \tilde{v}_j(\tilde{\pi}(j)) \geq \tilde{v}_j(T)] 
\\
&\geq (1-\tilde{\epsilon}) \mathbb{P}_{T\sim \mathcal{\widetilde{D}}}[\cup_{i\in T}  v_i(\tilde{\pi}(i)) \geq v_i(T)  | \cup_{j\in T} \tilde{v}_j(\tilde{\pi}(j)) \geq \tilde{v}_j(T)] 
\\
&\geq (1-\tilde{\epsilon}) \mathbb{P}[(\cup_{i\in T} v_i(\tilde{\pi}(i)) \geq v_i(T))  \cap (\cup_{j\in T} \tilde{v}_j(\tilde{\pi}(j)) \geq \tilde{v}_j(T))]
\\
&  ~~~~~ (\because \mathbb{P}(A|B) \geq \mathbb{P}(A\cap  B))
\\
&= (1-\tilde{\epsilon}) 
\mathbb{P}[\cup_{j\in T} \cup_{i\in T} \{v_i(\tilde{\pi}(i)) \geq v_i(T)  \cap \tilde{v}_j(\tilde{\pi}(j))  \geq \tilde{v}_j(T)\}]
\\
&\geq (1-\tilde{\epsilon}) \mathbb{P}[\cap_{i\in T} \{v_i(\tilde{\pi}(i)) \geq v_i(T)  \cap \tilde{v}_i(\tilde{\pi}(i))  \geq \tilde{v}_i(T)\}]  
\\
&= (1-\tilde{\epsilon}) \mathbb{P}[M(\tilde{\pi}, T)] = (1-\tilde{\epsilon}) f_T(\textit{\textbf{p}},\boldsymbol{\alpha}) = 1-\epsilon.
\end{aligned}
\end{equation*}
This ends the proof.
\end{proof}
In Theorem \ref{thm: prob_same_partition_n_partial_info}, we have $(1-\tilde{\epsilon}) f_T(\textit{\textbf{p}},\boldsymbol{\alpha}) = 1-\epsilon$ for any $\tilde{\epsilon}>0$. This implies $\epsilon = 1- (1-\tilde{\epsilon}) f_T(\textit{\textbf{p}},\boldsymbol{\alpha}) = \tilde{\epsilon}$ if $f_T(\textit{\textbf{p}},\boldsymbol{\alpha}) = 1$. So, for an arbitrary $\epsilon>0$, we have arbitrary $\tilde{\epsilon}>0$ if $f_T(\textit{\textbf{p}},\boldsymbol{\alpha}) = 1$. However, it is not true even in the $l=2$ support noise model. For example, as we see in Sec. \ref{sec: N_agent_2_support_partial_info} of the Appendix that we have $f_T(p,\alpha)=1$ iff $p=0$ or $p=1$, i.e., when values of all the coalitions are either scaled by some scalar $\alpha > 1$, or they are retained. Therefore, we relax the requirement of $f_T(\textit{\textbf{p}},\boldsymbol{\alpha}) =1$, and ask for $f_T(\textit{\textbf{p}},\boldsymbol{\alpha}) =\zeta$ for user-given $\zeta$. In some situations, the $\zeta$ captures the satisfaction value of an external agent trying to predict the partition of a noise-free game without having its knowledge. That is, a higher $\zeta$ is preferred. In particular, if $\zeta=1$, we have $\epsilon = \tilde{\epsilon}$. 
\begin{theorem}
\label{thm: zeta_noise_robust}
If a partition $\tilde{\pi}$ is $\tilde{\epsilon}$-PAC stable for the noisy game $(N, \tilde{\textbf{\textit{v}}})$ and for $\epsilon = 1 - (1 - \tilde{\epsilon})\zeta$ it is $\epsilon$-PAC stable for the noise-free game $(N, \textbf{\textit{v}})$, then it is also $\zeta$ noise-robust core-stable. 
\end{theorem}

\begin{proof}
Recall, from Theorem \ref{thm: prob_same_partition_n_partial_info}, we have $\mathbb{P}[T ~core~blocks~\tilde{\pi}~for~(N, \textbf{\textit{v}})] \leq \epsilon$. Here $\epsilon =  1- (1-\tilde{\epsilon})f_T(\textbf{\textit{p}}, \boldsymbol{\alpha})$. Setting $f_T(\textbf{\textit{p}}, \boldsymbol{\alpha}) = \zeta$, we have $\epsilon =  1- (1-\tilde{\epsilon})\zeta = 1 - \zeta + \tilde{\epsilon} \zeta$. So, for this $\epsilon$, the partition $\tilde{\pi}$ is $\zeta$ noise-robust core-stable from Definition \ref{def: zeta_noise_robustness}.
\end{proof}
The agreement probability $f_T(\textit{\textbf{p}},\boldsymbol{\alpha})$ being the same as user-given satisfaction value $\zeta$ identifies the noise regimes $I^{\star}(T, \zeta)$ for which at least the $\zeta$ fraction of preferences are preserved.  The following Theorem shows that the noise-regime for which a partition $\tilde{\pi}$ is core-stable in both noisy and noise-free games with a user-given satisfaction value $\zeta$ is indeed non-empty.

\begin{theorem}
\label{thm: noise_set_partial_info} 
Let $\tilde{\pi}$ be $\tilde{\epsilon}$-PAC stable partition of noisy game $(N, \tilde{\textbf{\textit{v}}})$ and
it is $\epsilon$-PAC stable for noise-free game $(N, \textbf{\textit{v}})$. Then, for a sample $\mathcal{S}_{t}=\{T_1, \dots, T_{m_{t}}\}$ drawn i.i.d. from $\mathcal{\widetilde{D}}$ we obtain a non-empty noise regime $I^{\star}(\mathcal{S}_t, \zeta) = \cap_{i=1}^{m_{t}} I^{\star}(T_i, \zeta)$ for which $\tilde{\pi}$ is $\zeta$ noise-robust core-stable partition. Moreover, $\tilde{\pi}$ is $\zeta$ noise-robust core-stable partition for the noise regime $I^{\star}(\zeta) = \cap_{T \subseteq N} I^{\star}(T, \zeta)$.
\end{theorem}
\begin{proof}
For any coalition $T$, we first note that $I^{\star}(T,\zeta) \neq \emptyset$, because $\alpha(\tilde{\pi}(i)) = 1, ~\forall \tilde{\pi}(i) \in \mathcal{R}(T); ~\alpha(T) =1$ is always an element of $M(\tilde{\pi}, T)$. 
So, for $\mathcal{S}_t = \{T_1, \dots, T_{m_{t}}\}$ we have non-empty noise regimes  $I^{\star}(T_1, \zeta), \cdots,  I^{\star}(T_{\tilde{m}_t}, \zeta)$. Also,  $\alpha(S) = 1, ~\forall~ S\subseteq N$ is a common element of each $I^{\star}(T, \zeta),~\forall ~T\in \mathcal{S}_t$. Therefore, $I^{\star}(\mathcal{S}_t, \zeta) = \cap_{i=1}^{m_{t}} I^{\star}(T_i, \zeta) \neq \emptyset$, i.e., is non-empty. 
Hence, partition $\tilde{\pi}$ is $\zeta$ noise robust on the sample $\mathcal{S}_t$ with noise regime $I^*(\mathcal{S}_t, \zeta)$ in accordance to Theorem \ref{thm: prob_same_partition_n_partial_info} and Definition \ref{def: zeta_noise_robustness}.  Moreover, $I^{\star}(\zeta) \neq \emptyset$ because of the same reason as mentioned above. 
The $\zeta$ noise-robustness follows from Theorem \ref{thm: prob_same_partition_n_partial_info} and \ref{thm: zeta_noise_robust}.
\end{proof}
In the next Section, we provide the relation between $m$, and $\tilde{m}$, i.e., the number of samples used to get $\epsilon$ and $\tilde{\epsilon}$-PAC stable partition $\tilde{\pi}$ in noise-free and noisy game, respectively for top-responsive hedonic games \citep{alcalde2004researching} and other hedonic games.

\subsection{Sample size for top-responsive and other games}
\label{subsec: sample_complexity}
In a top-responsive game, the value of each agent in a given coalition depends on the most preferred sub-coalition. Formally, the top-responsive games are described via choice sets $Ch(i, S)$, defined as $Ch(i,S) \coloneqq \{X \subseteq S : \forall ~Y\subseteq S, i\in Y: X \succeq_i Y\}$. The game satisfies the top-responsiveness if (a) $\forall~ i\in N$, and $S\in \mathcal{C}_i, |Ch(i,S)| = 1$, and (b) $\forall~ i\in N$, and $S, T\in \mathcal{C}_i$ if $Ch(i,S) \succ_i Ch(i,T)$ then $S\succ_i T$ or if $Ch(i,S) = Ch(i,T)$, and $S \subset T$, then $S\succ_i T$.
\begin{theorem}
\label{thm: m_tilde_m_relation}
For a top-responsive game, let $\tilde{m}$ be the number of samples required to get $\tilde{\epsilon}$-PAC stable partition in noisy game $(N, \tilde{\textbf{\textit{v}}})$, and $m$ be the samples required for $\tilde{\pi}$ to be $\epsilon$-PAC partition in \textit{unknown} noise-free game $(N, \textbf{\textit{v}})$. Then $m \zeta \leq 
\tilde{m} \leq m + (2n^3 + 2n^4)\left(  \frac{(1-\tilde{\epsilon}) + \tilde{\epsilon} \zeta}{\tilde{\epsilon} (1+ \tilde{\epsilon} \zeta)} \log \frac{2n^3}{\delta} \right).$
\end{theorem}
\begin{proof}
Recall, to get $\tilde{\epsilon}$-PAC stable partition in the noisy top-responsive games authors in \citep{sliwinski2017learning} provide $\tilde{m}$ for top-responsive games as $\tilde{m}  = (2n^3 + 2n^4)\left(\frac{1}{\tilde{\epsilon}} \log \frac{2n^3}{\delta} \right)$. 
However, from  Theorem \ref{thm: prob_same_partition_n_partial_info} we have $(1-\tilde{\epsilon}) \zeta = 1-\epsilon$, this implies $\epsilon = (1-\zeta) + \zeta \tilde{\epsilon} \geq \zeta \tilde{\epsilon}$. Thus, for $\epsilon$-PAC stability of partition $\tilde{\pi}$ in a top-responsive noise-free game  the number of samples $m$ are given by $m  = (2n^3 + 2n^4)\left(\frac{1}{\epsilon} \log \frac{2n^3}{\delta} \right) \leq(2n^3 + 2n^4)\left(\frac{1}{ \zeta \tilde{\epsilon}} \log \frac{2n^3}{\delta} \right) = \frac{\tilde{m}}{\zeta}$. This gives an upper bound. For a lower bound, again consider $(1-\tilde{\epsilon}) \zeta = 1-\epsilon$, therefore we have $\epsilon = (1-\zeta) + \zeta \tilde{\epsilon} \leq 1 + \zeta \tilde{\epsilon}$. That is $\frac{1}{\epsilon} \geq \frac{1}{1 + \zeta \tilde{\epsilon}}$. Thus, we have
\begin{align*}
m  &= (2n^3 + 2n^4)\left(\frac{1}{\epsilon} \log \frac{2n^3}{\delta} \right) \geq (2n^3 + 2n^4)\left( \frac{1}{1 + \zeta \tilde{\epsilon}} \log \frac{2n^3}{\delta} \right)
\\
&= (2n^3 + 2n^4)\left( \left\lbrace \frac{1}{\tilde{\epsilon}} - \frac{(1-\tilde{\epsilon}) + \tilde{\epsilon} \zeta}{\tilde{\epsilon} (1+ \tilde{\epsilon} \zeta)} \right\rbrace \log \frac{2n^3}{\delta} \right)
= \tilde{m} - (2n^3 + 2n^4)\left(  \frac{(1-\tilde{\epsilon}) + \tilde{\epsilon} \zeta}{\tilde{\epsilon} (1+ \tilde{\epsilon} \zeta)} \log \frac{2n^3}{\delta} \right).
\end{align*}
From the lower and upper bounds, we have the result.
\end{proof}
The above Theorem gives a bound on the extra samples required to get $\epsilon$-PAC stable partition of the unknown noise-free game given $\tilde{\epsilon}$-PAC stable partition of the noisy game. Again the number of samples to get $\epsilon = (1 - (1-\tilde{\epsilon})\zeta)$-PAC stable outcome in an \textit{unknown} noise-free game are bounded by the number of samples $\tilde{m}$,
the satisfaction value $\zeta$, and the confidence parameter $\delta$. In particular, the number of samples $m$ are polynomial in $n, \frac{1}{\epsilon}, \log \left(\frac{1}{\delta} \right)$, but its upper bound is non-linear in $\zeta$.

We next relate the number of samples and errors in noisy and \textit{unknown} noise-free games. Let $\tilde{\pi}$ be $\tilde{\epsilon}$-PAC stable partition of noisy game when $\tilde{m}$ samples are used. Suppose, we get $(\tilde{\epsilon} - \tilde{\epsilon}^{\prime})$-PAC partition of the noisy game on increasing the noisy samples to $\tilde{m} + \tilde{m}^{\prime}$. Let $\tilde{\pi}$ be $(\tilde{\epsilon} - \tilde{\epsilon}^{\prime})$-PAC stable for the noisy game that uses $\tilde{m} + \tilde{m}^{\prime}$ samples. Let $\epsilon_{new}$ be the error incurred to get $\tilde{\pi}$ partition with $\tilde{m} + \tilde{m}^{\prime}$ samples in a given noise-free game, then $\epsilon_{new} = 1- (1-(\tilde{\epsilon} - \tilde{\epsilon}^{\prime})) f_{T}(\textit{\textbf{p}},\boldsymbol{\alpha}) = 1- (1-\tilde{\epsilon}) f_{T}(\textit{\textbf{p}},\boldsymbol{\alpha}) - \tilde{\epsilon}^{\prime} f_{T}(\textit{\textbf{p}},\boldsymbol{\alpha}) = \epsilon - \tilde{\epsilon}^{\prime} f_{T}(\textit{\textbf{p}},\boldsymbol{\alpha}) \leq  \epsilon $.
\begin{theorem}
\label{thm: epsilon_new}
For an unknown noise-free game, let $\tilde{\pi}$ be $\epsilon_{new}$-PAC stable partition with $\tilde{m}  + \tilde{m}^{\prime}$ noisy samples, and it is $\epsilon$-PAC stable partition with $\tilde{m}$ noisy samples, then $\epsilon_{new} \leq \epsilon$.
\end{theorem}

\begin{remark}
The results of Theorem \ref{thm: m_tilde_m_relation} and Theorem \ref{thm: epsilon_new} can be generalized to any class of hedonic games by suitably obtaining the sample complexity of that class. This is because the number of samples required in noise-free game is function of $n, \frac{1}{\epsilon}, \log \left( \frac{1}{\delta} \right)$.
\end{remark}
To get some more insights we next identify the agreement probability $f_T(\textbf{\textit{p}}, 
\boldsymbol{\alpha})$ defined for partial information noise model with $l\geq 2$ noise support in the following subsection. We use the base case of $l=2$ noise support case in the proofs of results in the next Section.
these are deferred to Sec. \ref{sec: N_agent_2_support_partial_info} of the Appendix.

\subsection{$n$ agent $l$ support partial information noisy game}
\label{subsec: n_agent_l_support_partial_info} 
We now consider the $l\geq 2$ support case, i.e.,  $\mathcal{N}_{sp} =  \{\alpha_1,\alpha_2,\ldots, \alpha_l\}$ with respective probabilities $p_1,p_2,\ldots, p_l$, and $\sum_{j\in[l]}p_j=1 $. Here $p_j = \mathbb{P}(\alpha(S) = \alpha_j)$ and  $\alpha_j >0, ~\forall~j\in [l]$. Moreover, without loss of generality we assume that $\alpha_i < \alpha_j, ~\forall~ i< j$.  
For above noise support the following Theorem give expression of $f_T(\textit{\textbf{p}},\boldsymbol{\alpha})$. To this end, for any coalition $T$ and for all $r,s$ such that $\alpha_r > \alpha_s$, define $\mathcal{I}(\alpha_r,\alpha_s, T) = \left\lbrace \tilde{\pi}(i) \in \mathcal{R}(T) ~\bigg|~ \frac{\tilde{v}_i(\tilde{\pi}(i))}{\tilde{v}_i(T)} \geq \frac{\alpha_r}{\alpha_s} \right \rbrace$.
\begin{theorem}
\label{thm: n_agent_l_support_f_T_p}
Let $\tilde{\pi}$ be a $\tilde{\epsilon}$-PAC stable outcome of the noisy game $(N, \tilde{\textbf{\textit{v}}})$ and let $\tilde{\pi}$ be a $\epsilon$-PAC stable outcome of noise-free game $(N, \tilde{\textbf{\textit{v}}})$, where $\epsilon$ is identified as in Theorem \ref{thm: prob_same_partition_n_partial_info}. Then for noise support $\mathcal{N}_{sp} = \{\alpha_1, \alpha_2, \dots, \alpha_l\}$, the $f_T(\textbf{p}, \boldsymbol{\alpha})$ is given by: 
\begin{equation*}
f_T(\textit{\textbf{p}},\boldsymbol{\alpha})=\begin{cases} 1, \hfill if~ \tilde{\pi}(i)= T,
~\forall ~i\in T, 
\\
\sum_{r,s\in [l]:\alpha_r>\alpha_s} p_s^{|\mathcal{R}(T)| - |\mathcal{I}(\alpha_r,\alpha_s, T)| + 1}
\times \{(p_r + p_s)^{|\mathcal{I}(\alpha_r,\alpha_s,T)|} - p_s^{|\mathcal{I}(\alpha_r,\alpha_s, T)|}\}
\\
+\sum_{a=1}^{l} p_a \left(\sum_{b=1}^{a}p_b\right)^{|\mathcal{R}(T)|},  \hfill otherwise.
\end{cases}
\end{equation*}
\end{theorem}
The proof uses the principle of Mathematical induction on noise support $l\geq 2$ with base case of $l=2$ support (Lemma \ref{lemma: f_t(p,alpha)} of the SM). The detailed proof is available in Appendix \ref{proof: f_t_p_alpha_n_player_l_support}.
\begin{remark}
If we allow  $f_T(\textbf{p},\boldsymbol{\alpha})= \zeta$ for all coalitions $T\subseteq N$ for some user-given satisfaction value $\zeta$, we have noise set
$I^{\star}(\zeta)$ in accordance to Theorem \ref{thm: noise_set_partial_info}. This noise set corresponds to the superlevel sets of the prediction probability function.
For this super level set the partition $\tilde{\pi}$ is $\zeta$ noise-robust core-stable. Later, Sec. \ref{sec: 2 agent_full_info} shows that these superlevel sets are non-convex by explicitly deriving the prediction probability.
\end{remark}

\section{Partial information noisy game when $\tilde{\pi}$ is not $\tilde{\epsilon}$-PAC stable partition } 
\label{sec: Non-existence_parital_info}
So far we have assumed that $\tilde{\pi}$ is $\tilde{\epsilon}$-PAC stable partition of the noisy game $(N, \tilde{\textbf{\textit{v}}})$; however, that is not always the case. For example $\tilde{\pi} =  \{\{1\},\{23\}\}$ is not core stable for the game in \eqref{eqn: noisy_game_motivation}. In this section, we consider the other case where an estimated partition $\tilde{\pi}$ is not $\tilde{\epsilon}$-PAC stable for the noisy game $(N, \tilde{\textbf{\textit{v}}})$. Note that $\tilde{\pi}$ not being $\tilde{\epsilon}$-PAC stable doesn't mean that the noisy game $(N, \tilde{\textbf{\textit{v}}})$ has no stable partition.  Given a sample $\mathcal{\tilde{S}}$, we say  $\tilde{\pi}$ is not $\tilde{\epsilon}$-PAC stable partition of noisy game $(N, \tilde{\textbf{\textit{v}}})$ if there is a coalition $T$ that core blocks it with probability at least $1-\tilde{\epsilon}$. Formally, $\forall ~\tilde{\epsilon}>0, ~\exists~T\sim \mathcal{\widetilde{D}}$, such that
\begin{equation}
\label{eqn: no_core}
\mathbb{P} [\cap_{i\in T}~~\tilde{v}_i(T) > \tilde{v}_i(\tilde{\pi}(i))] \geq 1-\tilde{\epsilon}.
\end{equation} 
Our interest is in finding the prediction probability (Equation (\ref{eqn: predi_prob_with_core})) that a noise-free game does not have $\tilde{\pi}$ as $\epsilon$-PAC stable outcome ($\epsilon$ to be identified in terms of $\tilde{\epsilon}$) when the noisy game does not have  $\tilde{\pi}$ as $\tilde{\epsilon}$-PAC stable partition. To this end, for any coalition $T$, we again define an agreement event $F(T,\tilde{\pi})$ \footnote{
Though we use the same names, the agreement event in Sec. \ref{sec: n_agent_partial_info} is different from this agreement event.}. It contains all the noise values $(\alpha(T), \{\alpha(\tilde{\pi}(i)))\}_{\tilde{\pi}\in \mathcal{R}(T)})$ such that coalition $T$ is preferred over all the coalitions $\tilde{\pi}(i) \in \mathcal{R}(T)$ by every agent $i\in T$ in both the noisy and noise-free games. Formally, 
\begin{small}
\begin{equation*}
\label{eqn: F(T)}
F(T,\tilde{\pi}) \coloneqq \{(\alpha(T), \{\alpha(\tilde{\pi}(i))\}_{\tilde{\pi}(i)\in \mathcal{R}(T)}) : \cap_{i\in T} \{v_i(T) \geq v_i(\tilde{\pi}(i)) \cap \alpha(T){v}_i(T)  \geq \alpha(\tilde{\pi}(i)){v}_i(\tilde{\pi}(i))\}\}.
\end{equation*}
\end{small}
For probability mass $\textbf{\textit{p}}$ and noise value set $\boldsymbol{\alpha}$, let $h_T(\textbf{\textit{p}},\boldsymbol{\alpha}) \coloneqq \mathbb{P}_{T\sim \mathcal{\widetilde{D}}}[F(T, \tilde{\pi})]$ be the agreement probability. Note that $F(T, \tilde{\pi})$ and hence $h_T(\textbf{\textit{p}},\boldsymbol{\alpha})$ are not known since the noise-free values $v_i(T)$ and $v_i(\tilde{\pi}(i))$ are not known. However, for $l \geq 2$ support noise distribution $\mathcal{N}$ we obtain $h_T(\textbf{\textit{p}}, \boldsymbol{\alpha})$ explicitly in Sec. \ref{subsec: n_agents_l_support_non_existence}.
\begin{theorem}
\label{thm: unstab_unstable}
Suppose the noisy game $(N, \tilde{\textbf{\textit{v}}})$ does not have $\tilde{\pi}$ as $\tilde{\epsilon}$-PAC stable outcome, i.e., equation (\ref{eqn: no_core}) is satisfied. 
Then the prediction probability given in Equation \eqref{eqn: predi_prob_with_core} is given by: $\mathbb{P}[\cap_{i\in T}({v}_i(T) > {v}_i(\tilde{\pi}(i)))] \geq (1-\tilde{\epsilon}) h_T(\textbf{p}, \boldsymbol{\alpha})$, where $\epsilon > 0$ satisfy $
(1-\tilde{\epsilon}) h_T(\textbf{p}, \boldsymbol{\alpha}) = 1- \epsilon.$
\end{theorem}
\begin{proof}
Consider the following: $\mathbb{P}_{T\sim \mathcal{\widetilde{D}}} [\cap_{i\in T}~({v}_i(T) > {v}_i(\tilde{\pi}(i)))]$
\begin{equation*}
\begin{aligned}
&\geq \mathbb{P}_{T\sim \mathcal{\widetilde{D}}}[\cap_{i\in T}~({v}_i(T) > {v}_i(\tilde{\pi}(i))) ~|~ \cap_{i\in T}(\tilde{v}_i(T) > \tilde{v}_i(\tilde{\pi}(i)))] \times \mathbb{P}_{T\sim \mathcal{\widetilde{D}}}[\cap_{i\in T}(\tilde{v}_i(T) > \tilde{v}_i(\tilde{\pi}(i)))]
\\
&\geq (1-\tilde{\epsilon})~ \mathbb{P} [\cap_{i\in T}({v}_i(T) > {v}_i(\tilde{\pi}(i))) ~|~ \cap_{i\in T}(\tilde{v}_i(T) > \tilde{v}_i(\tilde{\pi}(i)))] 
\\
& \geq  (1-\tilde{\epsilon}) ~ \mathbb{P} [\cap_{i\in T}\{{v}_i(T) > {v}_i(\tilde{\pi}(i)) \cap \tilde{v}_i(T) > \tilde{v}_i(\tilde{\pi}(i))\}] \hspace{4 mm} (\because \mathbb{P}(A|B) \geq \mathbb{P}(A\cap B))
\\
&= (1-\tilde{\epsilon})~ h_T(\textit{\textbf{p}},\boldsymbol{\alpha}) = 1-\epsilon.
\end{aligned}
\end{equation*}
This ends the proof.
\end{proof}
Let $\eta$ be the probability of noise agreement event for which coalition $T$ core blocks $\tilde{\pi}$, i.e., $\eta \coloneqq h_T(\textit{\textbf{p}},\boldsymbol{\alpha})$. Thus  $\epsilon = (1-\eta) +  \eta \tilde{\epsilon}$ and hence partition $\tilde{\pi}$ is $\eta$ noise-robust non core-stable in accordance to Definition \ref{def: eta_noise_robustness}. Moreover, if $\eta = 1$ then $\tilde{\epsilon} = \epsilon$ so, for arbitrary $\tilde{\epsilon} > 0$, we also have arbitrary $\epsilon > 0$. 
\begin{remark}
\label{remark: non_existence_zeta}
Similar to Sec. \ref{sec: n_agent_partial_info}, for a user-given $\eta$, we get a noise set $I^{\star}(T, \eta)$ on $\textbf{p}$ for coalition $T$, i.e., the noise set in which the coalition $T$ core blocks  $\tilde{\pi}$ with error more than $1-\epsilon$. This is obtained by setting $\mathbb{P}[F(T, \tilde{\pi})] =  h_T(\textbf{\textit(p)}, \boldsymbol{\alpha}) = \eta$;  that is, $I^{\star}(T, \eta)$ is $\eta$ level set of agreement probability function $\mathbb{P}[F(T, \tilde{{\pi}})]$; in other words, it is a super level set of the prediction probability. 
Hence, $\tilde{\pi}$ is $\eta$ noise-robust non core-stable in this noise set $I^{\star}(T, \eta)$. 
\end{remark}
To better understand the noise robustness, we provide the expression of $h_T(\textit{\textbf{p}},\boldsymbol{\alpha})$ for $l\geq 2$ support noise models in the following subsection. For $l=2$ support noise model, we refer the readers to Lemma \ref{lemma: h_T(p,alpha)_2_support} of the SM. The detailed analysis of the $2$ support model gives many more insights and also serves as the base case in the proof of results in the next Section.

\subsection{$n$ agents $l$ support partial information noisy game without core}
\label{subsec: n_agents_l_support_non_existence}
In this section, we obtain the expression of the agreement probability $h_T(\textbf{\textit{p}}, \boldsymbol{\alpha})$ for $l\geq 2$ support noise model, $\mathcal{N}_{sp}  = \{\alpha_1,\alpha_2,\ldots, \alpha_l\}$. To this end, for all $r,s$ such that $\alpha_r > \alpha_s$ define $\mathcal{J}(\alpha_r,\alpha_s, T) = \left\lbrace \tilde{\pi}(i) \in \mathcal{R}(T) ~\bigg|~ \frac{\tilde{v}_i(\tilde{\pi}(i))}{\tilde{v}_i(T)} \geq \frac{\alpha_s}{\alpha_r} \right \rbrace$. It contains the set of all coalitions in the set $\mathcal{R}(T)$, such that $\alpha_r>\alpha_s$, and $\frac{\tilde{v}_i(\tilde{\pi}(i))}{\tilde{v}_i(T)} \geq \frac{\alpha_s}{\alpha_r}$. The following Theorem provides the expression of  $h_T(\textit{\textbf{p}},\boldsymbol{\alpha})$. For proof refer to Sec. \ref{proof: h_t_p_alpha_n_player_l_support} of the SM.
\begin{theorem}
\label{thm: n_agent_l_support_h_T_p}
For $n$ agent noisy hedonic game $(N, \tilde{\textbf{\textit{v}}})$ with $\mathcal{N}_{sp}  = \{\alpha_1,\alpha_2,\ldots, \alpha_l\}$, the agreement probability $h_T(\textbf{p}, \boldsymbol{\alpha})$ is given by:
\begin{equation*}
h_T(\textbf{p},\boldsymbol{\alpha})=\begin{cases} 1, \hfill if~ \tilde{\pi}(i)= T, ~\forall ~i\in T,
\\
\sum_{r,s\in [l]:\alpha_r>\alpha_s} p_r^{|\mathcal{R}(T)| - |\mathcal{J}(\alpha_r,\alpha_s, T)| + 1} \times \{(p_s + p_r)^{|\mathcal{J}(\alpha_r,\alpha_s,T)|} -
p_r^{|\mathcal{J}(\alpha_r,\alpha_s, T)|}\}
\\
+ \sum_{a=1}^{l} p_a \left(\sum_{b=a}^{l}p_b\right)^{|\mathcal{R}(T)|}, \hfill otherwise.
\end{cases}
\end{equation*}
\end{theorem}
\begin{remark}
Let $h_T(\textbf{p},\boldsymbol{\alpha})= \eta$ for some user-given satisfaction value $\eta$, we get a set of noise values in accordance to the Remark \ref{remark: non_existence_zeta}. In this case, the noise set depends on $|\mathcal{R}(T)|$, and $|\mathcal{J}(\alpha_r, \alpha_s, T)|,~ \forall~\alpha_r>\alpha_s$ for  coalition $T$. Also, the partition $\tilde{\pi}$ is $\eta$ noise-robust non core-stable in the noise set $I^{\star}(T, \eta)$.
\end{remark}
\begin{remark}
Theorem \ref{thm: unstab_unstable} provides the probability that $\tilde{\pi}$ is not $\epsilon$-PAC stable outcome in noise-free game $(N, {\textbf{\textit{v}}})$, when it is not $\tilde{\epsilon}$-PAC stable outcome in noisy game. 
Therefore, the probability that a noise-free game has $\tilde{\pi}$ as $\epsilon$-PAC stable outcome, given that the noisy game does not have $\tilde{\pi}$ as $\tilde{\epsilon}$-PAC stable outcome is compliment of the probability in Theorem \ref{thm: unstab_unstable}. 
\end{remark}
\section{2 agent full information model}
\label{sec: 2 agent_full_info}
This Section considers the complete information game with 2 agents. So, valuations on all coalitions are known in the noisy game; hence, a noisy core-stable partition is also known. Even though this Section is a particular case of previous sections, we get many valuable insights that enhance our understanding regarding noise robustness in noisy hedonic games. For example, we have this counter-intuitive fact that the prediction probability can be high if the noise value occurs with a high probability. A concrete illustration is that the prediction probability turns out to be 1 when values of all agents are inflated by $\alpha>1$ with probability 1; in fact, both games have the same preferences and hence identical partitions.
\subsection{2 support noise distribution}
\label{subsec: 2 support_full_info} 
Let the noise support be $\mathcal{N}_{sp} \in \{1,\alpha\}$ with $\alpha>1$, such that $\mathbb{P}[\mathcal{N}_{sp} = \alpha] = p = 1- \mathbb{P}[\mathcal{N}_{sp} =1]$.  Note that $\alpha>1$ is not a restrictive condition; even if we allow $\alpha<1$, we will get results similar to the ones presented below. Given a noisy game and its corresponding core-stable partition, we aim to find the prediction probability that a core-stable partition of the \textit{unknown} noise-free game is the same as a core-stable partition of a noisy game. Formally, for a user-given $\zeta \in (0,1]$, we find $\mathbb{P}[\pi = \tilde{\pi} ~| ~noisy~game] \geq \zeta$. 

We want to emphasize that the above prediction probability is the same as the one given in Equation \eqref{eqn: predi_prob_with_core}. Since, in a 2 agent complete information game, the noisy core-stable partition $\tilde{\pi}$ always exists; hence $\tilde{\epsilon} = 0$ in Theorem \ref{thm: prob_same_partition_n_partial_info}. 
So, with $f_T(\textbf{\textit{p}}, \boldsymbol{\alpha}) = \zeta$ we have $\mathbb{P}_{T\sim \mathcal{\widetilde{D}}}[T \text{ does not core blocks } \tilde{\pi} ~ in ~(N, \textbf{\textit{v}})] \geq \zeta$. Consider the following 2 agents' noisy game.
\begin{equation}
\label{eqn: game1}
\tag{game 1}
\tilde{v}_1(12) > \tilde{v}_1(1); \hspace{2 mm} \tilde{v}_2(12) > \tilde{v}_2(2).
\end{equation}
Clearly, $\tilde{\pi} = \{12\}=N$ is the core-stable outcome of the above noisy game. The following Lemma gives the prediction probability for the above game.
\begin{lemma}
\label{lemma: 2_agent_2_support_full_info}
For noisy \ref{eqn: game1} with complete information on $\tilde{\textbf{\textit{v}}}$, and $\mathcal{N}_{sp} = \{1,\alpha\}$, we have
\begin{equation}
\label{eqn: 2_support_prob}
\begin{aligned}
\mathbb{P}[\pi =\tilde{\pi}~|~game~1] = \begin{cases}
1-p(1-p^2), & if ~\alpha \geq \overline{r}
\\
1-p(1-p), & if~ \underline{r} \leq \alpha < \overline{r}
\\
1, & if~ \alpha < \underline{r},
\end{cases}
\end{aligned}
\end{equation}
where 
$\overline{r}= \max \left\lbrace \frac{\tilde{v}_1(12)}{\tilde{v}_1(1)}, \frac{\tilde{v}_2(12)}{\tilde{v}_2(2)} \right\rbrace$, and $ \underline{r} =\min \left\lbrace \frac{\tilde{v}_1(12)}{\tilde{v}_1(1)}, \frac{\tilde{v}_2(12)}{\tilde{v}_2(2)} \right\rbrace$. 

Also, this prediction probability $\mathbb{P}[\pi =\tilde{\pi}~|~game~1]$ is convex in $p$. So, while the minimal value for $\mathbb{P}[\pi =\tilde{\pi}~|~game~1]$ occurs for noise probabilities around $p = 0.5$ (depending on $\alpha,   \overline{r}$  and  $\underline{r}$), the maximal value of it is $1$ at $p = 0$ and $p =1$.
\end{lemma}
The proof is deferred to Sec. \ref{proof: 2_agent_2_support} of the SM. The above lemma has following insights: the prediction probability depends on three factors, $p$, $\alpha$, and $\tilde{\textbf{\textit{v}}}$. 
Note that $p$ is not known because the noise distribution is unknown. We know $\tilde{\textbf{\textit{v}}}$'s hence $\underline{r}$, and $\overline{r}$ are known. If $p$ is close to $0.5$, then the prediction probability $\mathbb{P}[\pi =\tilde{\pi}~|~game~1]$ is close to $0.62$, for $\alpha\geq \overline{r}$, and close to $0.75$ for $\underline{r} \leq \alpha < \overline{r}$. So, the prediction probability is the least when noise is random. We call this minimum prediction probability the \textit{safety value}. 

Suppose we allow the user-given satisfaction value on the prediction probability, i.e., we relax the condition $\mathbb{P}[\pi =\tilde{\pi}~|~game~1]=1$, and allow a user-given satisfaction value, $\zeta \in (0,1]$ on the  prediction probability, i.e.,
$\mathbb{P}[\pi =\tilde{\pi}~|~game~1]= \zeta$. So, for different ranges of $\alpha$, we get an interval of the noise probabilities $p$ allowable to attain a user-given probability $\zeta$. For example, if $\zeta = 0.9$, then the noise regime, $I^{\star}(\zeta = 0.9)  = [0,0.101] \cup [0.946,1]$, if $\alpha \geq \overline{r}$; it is $[0,0.113]\cup [0.887,1]$, if $\underline{r} \leq \alpha < \overline{r}$; and it is $[0,1]$, if $\alpha < \underline{r}$. So, for the noise set $I^{\star}(\zeta =0.9)$, the core-stable partition of the noise-free game is the same as the core-stable partition of the noisy \ref{eqn: game1} with probability $0.9$. Thus, the noise regime achieving a user given satisfaction value can be non-convex.
Figure \ref{fig: 2_agent_full_info} illustrates these observations.
\begin{figure}[h!]
\centering
\includegraphics[scale = 0.7]{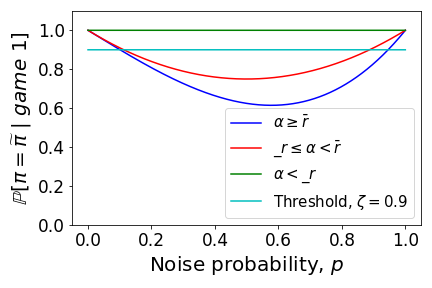}
\caption{In two agent hedonic \ref{eqn: game1} with 2 support noise model, we plot the prediction probability  $\mathbb{P}[\pi  = \tilde{\pi} ~|~ game~1]$ for different ranges of $\alpha$.}
\label{fig: 2_agent_full_info}
\end{figure}
Note that in Equation \eqref{eqn: 2_support_prob} we have obtained the conditional probability $\mathbb{P}[\pi =\tilde{\pi}~|~game~1]$. One can obtain a similar prediction probability for other noisy games, which we call game 2, game 3, and game 4, whose details are available in Sec. \ref{subsec: other_noisy_game_2_agent} of the SM.
We summarize the main observations of 2 agents 2 support noise in the Theorem below:
\begin{theorem}
\label{thm: 2_agents_2_support}
Consider 2 agent noisy hedonic game with 2 support noise model, then the prediction probability that $\pi = \tilde{\pi}$ given any noisy game $k$, $k=1,2,3,4$ is:
\begin{equation}
\label{eqn: grand_coaltition_noise_free_2_support}
\mathbb{P}[\pi =\tilde{\pi}~|~game~k] = \begin{cases} 1, & \text{under any condition in A},
\\
q(p,\tilde{v}_1(\cdot), \tilde{v}_2(\cdot), \alpha), & otherwise,
\end{cases}
\end{equation}
for some function $q(p,\tilde{v}_1(\cdot), \tilde{v}_2(\cdot), \alpha)< 1$, depending on $\tilde{v}_1(\cdot), \tilde{v}_2(\cdot)$ and $\alpha$. The conditions in A are (a) $k=1$, and $\alpha < \underline{r}$; (b) $k=2$ and $\frac{1}{\alpha} \geq \overline{r}$; (c) $k=3$ and $\frac{1}{\alpha} \geq \frac{\tilde{v}_1(12)}{\tilde{v}_1(1)} $; and (d) $k=4$ and $\frac{1}{\alpha} \geq \frac{\tilde{v}_2(12)}{\tilde{v}_2(2)}$. 

Moreover, if $p=0$ or $p=1$, we have $\mathbb{P}[\pi = \tilde{\pi}~|~game~k] =1$, for all $k=1,2,3,4$. 
\end{theorem}
For a 2 support noise model, the probability that a noise-free game has the same core-stable partition as the noisy game is $1$ in many cases, including $p=1$, i.e., when all values are inflated by $\alpha$. Thus, the allowable noise regimes for high prediction probabilities can include high noise values.
Moreover, $q(p,\tilde{v}_1(\cdot), \tilde{v}_2(\cdot), \alpha) \geq 0.62$ is the safety value. So, we have a lower bound on the prediction probability.

\subsection{3 support noise distribution}
\label{subsec: 3support_full_info}
Next, consider the 3 support noise,  $\mathcal{N}_{sp} = \{1,\alpha_1, \alpha_2\}$, where $\alpha_1 > 1$, and $0<\alpha_2 < 1$. Let $\mathbb{P}[\alpha(S) = \alpha_1] = p_1; ~ \mathbb{P}[\alpha(S) = \alpha_2] = p_2; ~ and~ \mathbb{P}[\alpha(S) = 1] = 1-p_1-p_2$. Given the noisy \ref{eqn: game1}, the prediction probability for 3 support noise is given in the Lemma below. 
\begin{lemma}
\label{lemma: 2_agent_3_support}
For the 3 support noise model the prediction probability $\mathbb{P}[\pi = \tilde{\pi} ~|~ game~1]$ is 
\begin{equation}
\mathbb{P}[\pi = \tilde{\pi} ~|~ game~1] = \begin{cases} g(p_1,p_2),  & ~if~~ \alpha_1 \geq \overline{r}~;~ \frac{1}{\alpha_2} \geq \overline{r} ~; ~ \frac{\alpha_1}{\alpha_2} \geq \overline{r}
\\
1, &~ if~~ \alpha_1 < \underline{r}~;~ \frac{1}{\alpha_2} < \underline{r} ~; ~ \frac{\alpha_1}{\alpha_2} < \underline{r}
\end{cases}
\end{equation}
where 
$g(p_1,p_2) = 
p_1^3 + p_2^3 +   2(p_1 (1-p_1-p_2)^2 +  p_2^2 (1-p_1-p_2) + p_1p_2(1-p_1-p_2) + p_1p_2^2 )    + p_1^2 p_2 + p_1^2 (1-p_1-p_2) + p_2 (1-p_1-p_2)^2 +(1-p_1-p_2)^3$.
\end{lemma}
There are $\mathbf{106}$ more cases in the above Lemma, where in each case, the prediction probability is strictly less than 1 (Sec. \ref{sec: proof_3_suppot_lemma} of SM). Unlike 2 support model (Sec. \ref{subsec: 2 support_full_info}) in this case they are \emph{non-convex}; a counter-example is available in Sec. \ref{subsec: counter_example_3_support} of Appendix

\section{Related work} 
\label{sec: related_work}
\textbf{Stability notions in hedonic games:} Researchers have extensively studied hedonic games in the computational social choice community. Some early works in coalition formation games describing the economic situations include that of \cite{dreze1980hedonic,elkind2009hedonic}. The agents collaborate and have personal preferences on different coalitions. Based on these preferences, agents seek a  partition of the agent set. However, which partition to form led to various notions of stability \citep{bogomolnaia2002stability, banerjee2001core, aziz2012existence}. Some of them are core stability, Nash stability, and perfect. In this work, we use core stability. In particular, the core for the simple hedonic games is available in \cite{banerjee2001core}.

\textbf{Representation of hedonic games:}
In many real-life scenarios, there are multiple agents, so storing the hedonic game in a machine takes exponential space. In literature, various concise representations are used because they (often) only require polynomial space. So, apart from various stability notions, much literature is on representing the hedonic games. Some of them includes individually rational lists of coalitions (IRLC) \citep{ballester2004np}, hedonic coalition nets (HCNs) \citep{elkind2009hedonic}, additively separable games, fractional hedonic games, $\mathcal{B}$-games, $\mathcal{W}$-games, top-responsive games \citep{alcalde2004researching}. A detailed survey of the hedonic games is available in \cite{Aziz_savani,aziz2019fractional, cechlarova2004stability}.
In our work, we only use partial information and  $\epsilon$-PAC stable notion for the existence of partition $\tilde{\pi}$. However, our work is valid for any class of hedonic games as long as both noise-free and noisy values have the same representation.

\textbf{Existence of solution concepts:} Another line of literature focuses on the algorithmic aspects of solution concepts of hedonic games. Regarding solution concepts like core stability and nature of partitions, there are two questions: does there exist a partition $\pi$ satisfying the solution concept's properties; if there is such a $\pi$, find one. To this end, for different classes of hedonic games, there are various algorithms and hardness results such as \cite{sung2010computational, rahwan2009anytime, woeginger2013core}. 

\textbf{PAC learning in hedonic games:} Uncertainty in the agents' preferences in the cooperative games has been carefully analyzed by \cite{balcan2015learning}. The authors used the PAC learning model 
to learn an underlying game. 
A new connection is established between PAC learnability and core stability for various classes of TU cooperative games.
It turned out that only a few classes of TU games are learnable and stable.  \cite{sliwinski2017learning} extended the PAC learning approach to the premise of hedonic games, where complete information about individual preferences is unavailable. We incorporate noise in the preferences and use PAC bounds to obtain the prediction probabilities.
\section{Discussion and looking ahead}
\label{sec: diss_looking_ahead}
This work considers the noisy hedonic game with partial information on preferences. Given a PAC stable partition of the noisy game, we find the prediction probability that \textit{unknown noise-free} game has PAC stable partition. This requires a combinatorial construct called agreement event and its probability. For $l\geq 2$ noise support, we obtain the agreement probability as a function of noise probabilities. For a user-given satisfaction value on agreement probability, we obtain the noise set such that a given partition is noise-robust. An interesting observation is that the prediction probability can be high for some high noise values with high probabilities. In particular, for a 2 agent game with $2$ noise support, we obtain the noise set for which the prediction probability is more than a user-given satisfaction value. 
We have noise robustness for the entire noise probability simplex for the prediction probability below $0.62$, i.e., the safety value. 
However, if the prediction probability function exceeds this safety value, the noise robust regime is non-convex. For the case of $3$ noise support, finding a safety value is difficult as it is a global minimum of a non-convex prediction probability function. We obtain the bounds on the extra noisy samples required to get the PAC stable partition in a noise-free game. These extra samples are polynomial in the number of agents and the user-given satisfaction value on agreement probability. 

The aspects we investigated offer many other rich possibilities; we mention some of them here. Firstly, since the prediction probability function for $3$ support noise distribution is non-convex, which renders the computation of the fundamental limit of noise robustness hard, one may investigate suitable approximations. Another possibility is to consider other noise models where the value of each coalition is perturbed at the individual player level. 



\acks{We would like to thank the anonymous Reviewers for their useful comments and suggestions. While working at this problem Prashant Trivedi was partially supported by the Teaching Assistantship offered by Government of India. Some part of this work was done when Nandyala Hemachandra was visiting IIM Bangalore on a sabbatical leave.}

\bibliography{acml22}

\newpage


\appendix
\vspace{-3mm}
\section{Proof of Theorem \ref{thm: n_agent_l_support_f_T_p}}
\label{proof: f_t_p_alpha_n_player_l_support}
\begin{proof}
We prove this via induction on noise support $l\geq 2$. The base case with $l=2$ support is available in Lemma \ref{lemma: f_t(p,alpha)} of the SM. Let us assume that it is true for $l=k$, i.e., there are sets $\mathcal{I}(\alpha_r,\alpha_s, T) = \left\lbrace \tilde{\pi}(i) \in \mathcal{R}(T) ~\bigg\vert~ \frac{\tilde{v}_i(\tilde{\pi}(i))}{\tilde{v}_i(T)} \geq \frac{\alpha_r}{\alpha_s} \right \rbrace$, such that $\alpha_s<\alpha_r, ~\forall~ 1 \leq s<r \leq k$.

For this $k$ we have $f_T(p_j, \alpha_j; j \in [k]) =: f_T(\textit{\textbf{p}},\boldsymbol{\alpha})$ (by assumption), here $[k] = \{1,2, \dots, k\}$ 
\begin{small}
\begin{equation*}
f_T(\textit{\textbf{p}},\boldsymbol{\alpha})  = \sum_{a=1}^{k} p_a \left(\sum_{b=1}^{a}p_b\right)^{|\mathcal{R}(T)|} + \sum_{r,s\in[k]:\alpha_r>\alpha_s} p_s^{|\mathcal{R}(T)| - |\mathcal{I}(\alpha_r,\alpha_s, T)| + 1} ((p_r + p_s )^{|\mathcal{I}(\alpha_r,\alpha_s,T)|}-p_s^{|\mathcal{I}(\alpha_r,\alpha_s, T)|}).
\end{equation*}
\end{small}
We will now show that this is true for $l=k+1$. To this end, for all $s\in [k]$ such that for $\alpha_{k+1}>\alpha_s$ we define $\mathcal{I}(\alpha_{k+1},\alpha_s, T) = \left\lbrace \tilde{\pi}(i) \in \mathcal{R}(T) ~\bigg\vert~ \frac{\tilde{v}_i(\tilde{\pi}(i))}{\tilde{v}_i(T)} \geq \frac{\alpha_{k+1}}{\alpha_s} \right \rbrace$. Now, there are two cases, $\mathcal{I}(\alpha_{k+1},\alpha_s, T) = \emptyset,~ \forall~ \alpha_s, ~s\in [k]$, or $\mathcal{I}(\alpha_{k+1},\alpha_s, T) \neq \emptyset$ for at least for one $s\in [k]$.

\textbf{Case 01:} [$\mathcal{I}(\alpha_{k+1},\alpha_s, T) = \emptyset, ~ \forall~ \alpha_s, ~s\in [k]$]. With one more element in noise support, apart from the existing  $\{\alpha(\tilde{\pi}(i)\}_{\tilde{\pi}(i) \in \mathcal{R}(T)}$, and $\alpha(T)$ for $k$ support case it will also have $\alpha(T) = \alpha_{k+1}$, and $\alpha(\tilde{\pi}(i)) \in \{\alpha_1, \alpha_2, \dots, \alpha_{k+1}\},~ \forall~ \tilde{\pi}(i) \in \mathcal{R}(T)$. The probability of such $\alpha$'s is $p_{k+1}\left(\sum_{b=1}^{k+1} p_b\right)^{|\mathcal{R}(T)|}$. Therefore, the overall probability is 
\begin{small}
\begin{equation*}
\sum_{a=1}^{k} p_a \left(\sum_{b=1}^{a}p_b\right)^{|\mathcal{R}(T)|} + p_{k+1}\left(\sum_{b=1}^{k+1} p_b\right)^{|\mathcal{R}(T)|}  = \sum_{a=1}^{k+1} p_a \left(\sum_{b=1}^{a}p_b\right)^{|\mathcal{R}(T)|}.
\end{equation*}
\end{small}

\textbf{Case 02:} [$\mathcal{I}(\alpha_{k+1},\alpha_s, T) \neq \emptyset$ for at least for one $s\in [k]$]. 
In this case, apart from the existing $\{\alpha(\tilde{\pi}(i))\}_{\tilde{\pi}(i) \in \mathcal{I}(\alpha_{r},\alpha_s, T)}$, and  $\alpha(T)$ for $k$ support, we also have $\{\alpha(\tilde{\pi}(i))\}_{\tilde{\pi}(i) \in \mathcal{I}(\alpha_{k+1},\alpha_s, T)}$, $\alpha(T)$ such that $\alpha(\tilde{\pi}(i)) = \alpha_s,~ \forall~\tilde{\pi}(i)\in \mathcal{R}(T) \setminus \mathcal{I}(\alpha_{k+1}, \alpha_s, T)$, and $\alpha(T) = \alpha_{k+1}$. Thus, for $k+1$  support the probability is:
\begin{small}
\begin{eqnarray*}
& & \sum_{r,s\in[k]:\alpha_r>\alpha_s} p_s^{|\mathcal{R}(T)| - |\mathcal{I}(\alpha_r,\alpha_s, T)| + 1} ((p_r + p_s)^{|\mathcal{I}(\alpha_r,\alpha_s,T)|}-p_s^{|\mathcal{I}(\alpha_r,\alpha_s, T)|}) 
\\
& & \hspace{15mm} +~ p_s^{|\mathcal{R}(T)| - |\mathcal{I}(\alpha_{k+1},\alpha_s, T)| + 1} ((p_{k+1} + p_s)^{|\mathcal{I}(\alpha_{k+1},\alpha_s,T)|}-p_s^{|\mathcal{I}(\alpha_{k+1}, \alpha_s, T)|}).
\end{eqnarray*}
\end{small}
From case 01 and case 02 above, for $k+1$ support we have,
\begin{small}
\begin{align*}
f_T(p_j, \alpha_j; j \in [k+1]) &= \sum_{r,s\in[k+1]:\alpha_r>\alpha_s} p_s^{|\mathcal{R}(T)| - |\mathcal{I}(\alpha_r,\alpha_s, T)| + 1} \left((p_r + p_s)^{|\mathcal{I}(\alpha_r,\alpha_s,T)|}-p_s^{|\mathcal{I}(\alpha_r,\alpha_s, T)|}\right)
\\
&  ~~+\sum_{a=1}^{k+1} p_a \left(\sum_{b=1}^{a}p_b\right)^{|\mathcal{R}(T)|}   
\end{align*}
\end{small}
Therefore, from the principle of Mathematical induction, this is true for any $l \geq 2$.
\end{proof}

\subsection{$n$ agents $2$ support partial information noise model}
\label{sec: N_agent_2_support_partial_info}
In a two support noise model we have $\mathcal{N}_{sp} = \{1,\alpha\}$ with $\alpha>1$, such that for any coalition $S\subseteq N$,  $\mathbb{P}[\alpha(S) = \alpha] = p = 1-\mathbb{P}[\alpha(S) = 1]$. We derive the agreement probability, $f_T(p, \alpha)$ in the following lemma. Note that this lemma serves as the base case in the Mathematical induction based proof of the Theorem \ref{thm: n_agent_l_support_f_T_p} in the main paper.

\begin{lemma}
\label{lemma: f_t(p,alpha)}
Let $\tilde{\pi}$ be $\tilde{\epsilon}$-PAC stable partition of noisy game $(N, \tilde{\textbf{\textit{v}}})$, and let $\tilde{\pi}$ be a $\epsilon$-PAC stable outcome of the noise-free game $(N, \textbf{\textit{v}})$, where $\epsilon$ is identified in Theorem \ref{thm: prob_same_partition_n_partial_info} of the paper. Then the agreement probability $f_T({p}, {\alpha})$ is given by
\begin{equation*}
f_T({p}, {\alpha}) = \begin{cases} 1, &if ~\tilde{\pi}(i) =T, ~\forall~i\in T
\\
p + (1-p)^{|\mathcal{R}(T)| + 1 - |\mathcal{I}(\alpha, T)|},  & otherwise
\end{cases}
\end{equation*}
where $\mathcal{I}(\alpha, T) = \left\lbrace \tilde{\pi}(i) \in \mathcal{R}(T) ~\bigg|~ \frac{\tilde{v}_i(\tilde{\pi}(i))}{\tilde{v}_i(T)} \geq \alpha \right\rbrace.$
\end{lemma}
\begin{proof}
Recall from Theorem \ref{thm: prob_same_partition_n_partial_info} in main paper we have the following 
\begin{equation*}
\mathbb{P}_{T\sim \mathcal{\tilde{D}}}[\cup_{i\in T}~ v_i(\tilde{\pi}(i)) \geq v_i(T) ] \geq (1-\tilde{\epsilon}) f_T(\textit{\textbf{p}}, \boldsymbol{\alpha}).
\end{equation*}
Also, recall that the agreement event is defined as 
\begin{small}
\begin{equation}
M(\tilde{\pi}, T) \coloneqq \{( \{\alpha(\tilde{\pi}(i))\}_{\tilde{\pi}(i)\in \mathcal{R}(T)}, \alpha(T)) : \cap_{i\in T} \{v_i(\tilde{\pi}(i)) \geq v_i(T)  ~\cap~ \alpha(\tilde{\pi}(i))  {v}_i(\tilde{\pi}(i))  \geq \alpha(T){v}_i(T)\}\}, \nonumber
\end{equation}
\end{small}
and $f_T({p}, {\alpha})  = \mathbb{P}[M(\tilde{\pi}, T)]$ is the probability of agreement event. Moreover,
\begin{equation*}
\mathcal{R}(T) \coloneqq \{\tilde{\pi}(i) ~|~ i\in T\}; \hspace{3 mm}
\mathcal{I}(\alpha, T) = \left\lbrace \tilde{\pi}(i) \in \mathcal{R}(T) ~\bigg|~ \frac{\tilde{v}_i(\tilde{\pi}(i))}{\tilde{v}_i(T)} \geq \alpha \right\rbrace.
\end{equation*}
To find the agreement probability, $f_T(\textbf{\textit{p}}, \boldsymbol{\alpha})$ we consider two cases $\mathcal{I}(\alpha, T) = \emptyset$, and $\mathcal{I}(\alpha, T) \neq \emptyset$. For these cases we identify the possible noise values $\{\alpha(\tilde{\pi}(i))\}_{\tilde{\pi}(i) \in \mathcal{R}(T)}$, $\alpha(T)$  that are element of $M(\tilde{\pi}, T)$.
\begin{itemize}
\item \textbf{Case 01:} $[\mathcal{I}(\alpha, T) = \emptyset]$. In this case, we have following elements in $M(\tilde{\pi}, T)$.
\begin{itemize}
\item $\alpha(\tilde{\pi}(i)) = 1,~ \forall~ \tilde{\pi}(i)\in \mathcal{R}(T)$ and $ \alpha(T) = 1$. The probability of such choice of $\alpha$'s is
\begin{equation}
\label{eqn: ft_part1_2_support}
    (1-p)^{|\mathcal{R}(T)| + 1}.
\end{equation}
\item $\alpha(\tilde{\pi}(i)) = \alpha$ for \textbf{exactly  one} $\tilde{\pi}(i)\in \mathcal{R}(T)$, and $\alpha(\tilde{\pi}(i)) = 1$ for remaining coalitions in $\mathcal{R}(T)$, and $\alpha(T) = \alpha$. Probability of such choice of $\alpha$'s is $(p \times (1-p)^{|\mathcal{R}(T)|- 1}) \times p$. And there are $|\mathcal{R}(T)| \choose 1$ ways of selecting \textbf{exactly one} coalition $\tilde{\pi}(i) \in \mathcal{R}(T)$. Thus, the probability of above $\alpha$'s is ${|\mathcal{R}(T)| \choose 1} p(1-p)^{|\mathcal{R}(T)| - 1} p$.

In general, for any $k \in \{0, 1, \dots, |\mathcal{R}(T)|\}$ coalitions $\tilde{\pi}(i) \in \mathcal{R}(T)$, take $\alpha(\tilde{\pi}(i))= \alpha$. Moreover, $\alpha(\tilde{\pi}(i))= 1$ for remaining $|\mathcal{R}(T)| - k$ coalitions and take $ \alpha(T) = \alpha$. Further, we have ${|\mathcal{R}(T)| \choose k}$ similar choices. So, the probability of the above choice of $\alpha$'s is 
\begin{small}
\begin{equation}
\label{eqn: ft_part2_2_support}
\begin{aligned}
\sum_{k=0}^{|\mathcal{R}(T)|} \left\lbrace{|\mathcal{R}(T)| \choose k} p^k(1-p)^{|\mathcal{R}(T)| - k} \right\rbrace \times p &= 
p \times \left(\sum_{k=0}^{|\mathcal{R}(T)|} {|\mathcal{R}(T)| \choose k} p^k(1-p)^{|\mathcal{R}(T)| - k}  \right) 
\\
&= p.
\end{aligned}
\end{equation}
\end{small}
This is because for any coalition $S$, we have $\mathbb{P}[\alpha(S) = \alpha] =p = 1-\mathbb{P}[\alpha(S) =1]$ and the fact that binomial probabilities summed up to 1.
\end{itemize}

\item \textbf{Case 02:} $[\mathcal{I}(\alpha, T) \neq \emptyset]$. Then, in addition to the above possible cases, we will have a few other cases, which are:
\begin{itemize}
\item $\alpha(\tilde{\pi}(i)) = \alpha$ for \textbf{exactly one} $\tilde{\pi}(i) \in \mathcal{I}(\alpha, T)$, $\alpha(\tilde{\pi}(i)) = 1$ for remaining coalitions in $\mathcal{R}(T)$ and $\alpha(T) = 1$. Probability of such choice of $\alpha$'s is $p(1-p)^{|\mathcal{R}(T)| - 1} (1- p)  = p(1-p)^{|\mathcal{R}(T)|}$. And there are $|\mathcal{I}(\alpha,T)| \choose 1$ ways of choosing \textbf{exactly one} coalition $\tilde{\pi}(i) \in \mathcal{I}(\alpha, T) $. Thus the overall probability is ${|\mathcal{I}(\alpha,T)| \choose 1} p(1-p)^{|\mathcal{R}(T)|} $.

In general, we have $\alpha(\tilde{\pi}(i)) = \alpha$ for any $k \in \{1, 2, \dots, |\mathcal{I}(\alpha,T)| \}$ coalitions $\tilde{\pi}(i) \in \mathcal{I}(\alpha, T)$. Moreover, $\alpha(\tilde{\pi}(i)) = 1$ for remaining $|\mathcal{R}(T)| - k$ coalitions, and $\alpha(T) = 1$. Probability of such choice of $\alpha$'s is $p^k (1-p)^{|\mathcal{R}(T)| - |\mathcal{I}(\alpha,T)|} (1- p)$. And there are $|\mathcal{I}(\alpha,T)| \choose k$ ways of selecting $k$ coalitions $\tilde{\pi}(i) \in \mathcal{I}(\alpha, T)$. Thus the overall probability is 
\begin{equation}
\label{eqn: ft_part3_2_support}
\sum_{k=1}^{|\mathcal{I}(\alpha,T)|} {|\mathcal{I}(\alpha,T)| \choose k} p^k (1-p)^{|\mathcal{R}(T)| - |\mathcal{I}(\alpha,T)|} (1-p).
\end{equation}
\end{itemize}
\end{itemize}
The probability of event $M(\tilde{\pi}, T)$, i.e., $\mathbb{P}[M(\tilde{\pi}, T)]$ is obtained by adding probabilities given in Equations \eqref{eqn: ft_part1_2_support}, \eqref{eqn: ft_part2_2_support} and \eqref{eqn: ft_part3_2_support}.
\begin{eqnarray*}
\mathbb{P}[M(\tilde{\pi}, T)]
&=& (1-p)^{|\mathcal{R}(T)| + 1}  + p + \sum_{k=1}^{|\mathcal{I}(\alpha,T)|} {|\mathcal{I}(\alpha,T)| \choose k} p^k (1-p)^{|\mathcal{R}(T)| - |\mathcal{I}(\alpha,T)|} (1-p)
\\
&=& (1-p)^{|\mathcal{R}(T)| + 1} + p + (1-p)^{|\mathcal{R}(T)| - |\mathcal{I}(\alpha,T)| + 1} \Bigg[ 1 -  (1-p)^{|\mathcal{I}(\alpha, T)|} \Bigg ]
\\
&=& p + (1-p)^{|\mathcal{R}(T)| - |\mathcal{I}(\alpha,T)| + 1}.
\end{eqnarray*}
This ends the proof.
\end{proof}



If $\tilde{\pi}(i) \neq T$ for at least one $i\in T$, then  $f_T({p}, {\alpha}) = 1,~ \forall~\alpha$ if and only if $p=0$ or $p=1$. 
That is, if the value of all the coalitions are retained, or if values
of all of them are inflated by $\alpha$, then for all $i\in T$, and for all $ \tilde{\pi}(i)\in \mathcal{R}(T)$,  one has $\tilde{\pi}(i) \succeq_i T$, and $\tilde{\pi}(i) \succeq_i^{\prime}  T$. Thus,  $\tilde{\pi}$ is $\epsilon$-PAC stable outcome of \textit{unknown noise-free} game and hence $\tilde{\pi}$ is noise-robust.

\begin{corollary}
When $\tilde{\pi} = N$, i.e., the grand coalition is $\tilde{\epsilon}$-PAC stable outcome in the noisy game, then $\mathcal{R}(T) = \{N\}$ for any coalition $T$. Thus, $\mathcal{I}(\alpha, T) = \emptyset$, or $\mathcal{I}(\alpha, T) = \{N\}$. Therefore, $f_T(p,\alpha)$ simplifies to 
\begin{equation}
f_T({p}, {\alpha}) = \begin{cases} 1, & ~if~ \mathcal{I}(\alpha, T) = \{N\}
\\
(1-p)^2 + p, & ~if~ \mathcal{I}(\alpha, T) = \emptyset.
\end{cases}
\end{equation}
\end{corollary}

\subsection{$n$ agents $2$ support partial information noisy games without core}
\label{sec: n_agents_2_support_non_existence}
Suppose $\tilde{\pi}$ is not $\tilde{\epsilon}$-PAC stable partition fo the noisy game $(N, \tilde{\textbf{\textit{v}}})$. Moreover, let the noise support be $\mathcal{N}_{sp} = \{1,\alpha\}$, the following lemma provides the expression of $h_T({p},{\alpha})$. Note that this lemma serves as the base case for the Mathematical induction based proof of Theorem \ref{thm: n_agent_l_support_h_T_p} in the main paper.

\begin{lemma}
\label{lemma: h_T(p,alpha)_2_support}
Suppose $\tilde{\pi}$ is not a $\tilde{\epsilon}$-PAC stable outcome of the noisy game $(N, \tilde{\textbf{\textit{v}}})$, then the agreement probability $h_T(p,\alpha)$ for noise support $\mathcal{N}_{sp} \in \{1,\alpha\}$ is given by
\begin{equation}
h_T({p},{\alpha}) = \begin{cases} 1, & if ~\tilde{\pi}(i) =T, ~\forall~i\in T
\\
(1-p) + p^{|\mathcal{R}(T)| + 1 - |\mathcal{J}(\alpha, T)|},  & otherwise,
\end{cases}
\end{equation}
where $\mathcal{J}(\alpha, T) \coloneqq \left\lbrace \tilde{\pi}(i) \in \mathcal{R}(T) ~\bigg|~ \frac{\tilde{v}_i(\tilde{\pi}(i))}{\tilde{v}_i(T)} \geq \frac{1}{\alpha} \right\rbrace.$
\end{lemma}

\begin{proof}
From Theorem \ref{thm: unstab_unstable} of the main paper, we have the following
\begin{equation*}
\mathbb{P}[\cup_{i\in T} v_i(\tilde{\pi}(i)) \geq v_i(T) ] \geq   (1-\tilde{\epsilon}) h_T(\textit{\textbf{p}},\boldsymbol{\alpha}).
\end{equation*}
To get $h_T({p}, {\alpha})  \coloneqq \mathbb{P}[F(T, \tilde{\pi})]$ we consider two cases viz. $\mathcal{J}(\alpha, T) = \emptyset$, and $\mathcal{J}(\alpha, T) \neq \emptyset$. For these cases, we identify the possible noise values elements of $F(T, \tilde{\pi})$.

\begin{itemize}
\item \textbf{Case 01:} [$\mathcal{J}(\alpha, T) = \emptyset$]. In this case, we have the following possibilities:
\begin{itemize}
\item $\alpha(\tilde{\pi}(i)) = \alpha,~ \forall~ \tilde{\pi}(i)\in \mathcal{R}(T)$, and $ \alpha(T) = \alpha$. Probability of such a choice of $\alpha$'s is
\begin{equation}
\label{eqn: ht_part1_2_support}
p^{|\mathcal{R}(T)| + 1}.
\end{equation}
\item $\alpha(\tilde{\pi}(i)) = 1$ for $k\in \{0, 1, \dots, |\mathcal{R}(T)|\}$ coalitions $\tilde{\pi}(i)\in \mathcal{R}(T)$, and $\alpha(\tilde{\pi}(i)) = \alpha$ for remaining $|\mathcal{R}(T)| - k$ coalitions. Moreover, $\alpha(T) = 1$. Probability of such choice of $\alpha$'s is $(1-p)^{k} p^{|\mathcal{R}(T)| - k} (1-p)$. Further, there are $|\mathcal{R}(T)| \choose k$ ways of selecting $k$ coalitions  $\tilde{\pi}(i)$ from $\mathcal{R}(T)$. Thus, the overall probability is 
\begin{equation}
\label{eqn: ht_part2_2_support}
\sum_{k=0}^{|\mathcal{R}(T)|} {|\mathcal{R}(T)| \choose k} (1-p)^{k} p^{|\mathcal{R}(T)| - k} (1-p) = 1-p.
\end{equation}
\end{itemize}

\item \textbf{Case 02:} [$\mathcal{J}(\alpha, T) \neq \emptyset$]. In addition to the above possible cases, we have a few other cases:
\begin{itemize}
\item $\alpha(\tilde{\pi}(i)) = 1$ for 
any $k\in \{1, 2, \dots, |\mathcal{J}(\alpha, T)| \}$ coalitions $\tilde{\pi}(i) \in \mathcal{J}(\alpha, T)$. Moreover,  $\alpha(\tilde{\pi}(i)) = \alpha$ for remaining coalitions in $\mathcal{R}(T)$. Also, $\alpha(T) = \alpha$. Probability of such choice of $\alpha$'s is $(1-p)^k p^{|\mathcal{R}(T)| - k} p  = (1-p)^{k} p^{|\mathcal{R}(T)|-k+1} $. And there are $|\mathcal{J}(\alpha,T)| \choose k$ ways of selecting $k$ coalitions $\tilde{\pi}(i) \in \mathcal{J}(\alpha,T)$. Thus the overall probability is
\begin{equation}
\label{eqn: ht_part3_2_support}
\sum_{k=1}^{|\mathcal{J}(\alpha, T)|}{|\mathcal{J}(\alpha,T)| \choose k} (1-p)^{k} p^{|\mathcal{R}(T)|-k+1}.  
\end{equation}
\end{itemize}
\end{itemize}
The probability $\mathbb{P}[F(T, \tilde{\pi})]$ is obtained by adding probabilities given in Equations \eqref{eqn: ht_part1_2_support}, \eqref{eqn: ht_part2_2_support} and \eqref{eqn: ht_part3_2_support}.
\begin{eqnarray*}
\mathbb{P}[F(T, \tilde{\pi})]
&=& p^{|\mathcal{R}(T)| + 1} + (1-p) + \sum_{k=1}^{|\mathcal{J}(\alpha, T)|}{|\mathcal{J}(\alpha,T)| \choose k} (1-p)^{k} p^{|\mathcal{R}(T)|-k+1}
\\
&=& p^{|\mathcal{R}(T)| + 1} + (1-p) + p^{|\mathcal{R}(T)| - |\mathcal{J}(\alpha,T)| + 1} \Bigg[ 1 - p^{|\mathcal{J}(\alpha, T)|} \Bigg ]
\\
&=& (1-p) + p^{|\mathcal{R}(T)| - |\mathcal{J}(\alpha,T)| + 1}.
\end{eqnarray*}
This ends the proof.
\end{proof}


If $\tilde{\pi}(i) \neq T$ for at least one $i\in T$, then  $h_T({p}, {\alpha}) = 1,~ \forall~\alpha$ if $p=0$ or $p=1$. That is, if the value of all coalitions are retained, or if value of all of them are inflated by $\alpha$, then coalition $T \succeq_i \tilde{\pi}(i) $, and $T \succeq_i^{\prime} \tilde{\pi}(i)$ for all $i\in T$. Thus, neither noise-free nor noisy game will have $\tilde{\pi}$ as PAC stable outcome. Moreover, if we allow $h_T(p,\alpha)= \eta$ for some user-given satisfaction $\eta$, we get a noise set in accordance to the Remark \ref{remark: non_existence_zeta} in the main paper. In this case, the noise set also depends on $|\mathcal{R}(T)|$, and $|\mathcal{J}(\alpha, T)|$ for coalition $T$. Hence, the partition is $\eta$ noise-robust non core-stable for the noise set $I^{\star}(T, \eta)$.

\subsection{Proof of Theorem \ref{thm: n_agent_l_support_h_T_p} }
\label{proof: h_t_p_alpha_n_player_l_support}

\begin{proof}
We will prove this via Mathematical induction on the noise support $l\geq 2$. Clearly, this is true for $l=2$ (from Lemma \ref{lemma: h_T(p,alpha)_2_support} above). Let us assume that it is true for $l=k$, i.e.; there are sets 
\begin{equation*}
\mathcal{J}(\alpha_r,\alpha_s, T) = \left\lbrace \tilde{\pi}(i) \in \mathcal{R}(T) ~\bigg\vert~ \frac{\tilde{v}_i(\tilde{\pi}(i))}{\tilde{v}_i(T)} \geq \frac{\alpha_s}{\alpha_r} \right \rbrace,
\end{equation*} 
such that the support $\alpha(S)= \{\alpha_1,\dots,\alpha_k\}, ~\forall~ S\subseteq N$ where $\alpha_s<\alpha_r, ~ \forall~ 1\leq s<r \leq k$. For this $k$ we have $f_T(p_j,\alpha_j: j\in [k]) =: h_T(\textit{\textbf{p}}, \boldsymbol{\alpha}) $ (by assumption)
\begin{small}
\begin{equation*}
h_T(\textit{\textbf{p}}, \boldsymbol{\alpha})  = \sum_{a=1}^{k} p_a \left(\sum_{b=a}^{k}p_b\right)^{|\mathcal{R}(T)|} + \sum_{r,s\in[k]:\alpha_r>\alpha_s} p_r^{|\mathcal{R}(T)| - |\mathcal{J}(\alpha_r,\alpha_s, T)| + 1} ((p_r + p_s )^{|\mathcal{J}(\alpha_r,\alpha_s,T)|}-p_r^{|\mathcal{J}(\alpha_r,\alpha_s, T)|}).
\end{equation*}
\end{small}
We will now show that this is true for $l=k+1$. To this end define $\mathcal{J}(\alpha_{k+1},\alpha_s, T)$  for all $s\in [k]$ such that $\alpha_{k+1}>\alpha_s$
\begin{equation*}
\mathcal{J}(\alpha_{k+1},\alpha_s, T) = \left\lbrace \tilde{\pi}(i) \in \mathcal{R}(T) ~\bigg\vert~ \frac{\tilde{v}_i(\tilde{\pi}(i))}{\tilde{v}_i(T)} \geq \frac{\alpha_{s}}{\alpha_{k+1}} \right \rbrace.
\end{equation*}
Now, there are two cases, $\mathcal{J}(\alpha_{k+1},\alpha_s, T) = \emptyset,~ \forall~ \alpha_s, ~s\in [k]$, or $\mathcal{J}(\alpha_{k+1},\alpha_s, T) \neq \emptyset$ for at least for one $s\in [k]$.

\textbf{Case 01:} [$\mathcal{J}(\alpha_{k+1},\alpha_s, T) = \emptyset,~ \forall~ \alpha_s, ~s\in [k]$]. Apart from the existing $\{\alpha(\tilde{\pi}(i)\}_{\tilde{\pi}(i) \in \mathcal{R}(T)}$ and  $\alpha(T)$ for $k$ support case, with this extra $k+1$, it will also have $\alpha(T) = \alpha_{k+1}$ and $\alpha(\tilde{\pi}(i))= \alpha_{k+1},~ \forall~ \tilde{\pi}(i) \in \mathcal{R}(T)$. The probability of such extra $\alpha$'s is 
$p_{k+1}\left(\sum_{b=k+1}^{k+1} p_b\right)^{|\mathcal{R}(T)|}$. Therefore, the overall probability is 
$$
\sum_{a=1}^{k} p_a \left(\sum_{b=a}^{k}p_b\right)^{|\mathcal{R}(T)|} + p_{k+1}\left(\sum_{b=k+1}^{k+1} p_b\right)^{|\mathcal{R}(T)|}  = \sum_{a=1}^{k+1} p_a \left(\sum_{b=a}^{k+1}p_b\right)^{|\mathcal{R}(T)|}.
$$

\textbf{Case 02:} [$\mathcal{J}(\alpha_{k+1},\alpha_s, T) \neq \emptyset$ for at least one $s\in [k]$]. In this case, apart from all $\alpha(T)$ and  $\{\alpha(\tilde{\pi}(i))\}_{\tilde{\pi}(i) \in \mathcal{J}(\alpha_{r},\alpha_s, T)}$, we have $\{\alpha(\tilde{\pi}(i))\}_{\forall~\tilde{\pi}(i) \in \mathcal{J}(\alpha_{k+1},\alpha_r, T)}, \alpha(T)$. For this set, the possible pairs are such that $\alpha(\tilde{\pi}(i)) = \alpha_r,~~ \forall~\tilde{\pi}(i)\in \mathcal{R}(T) \setminus \mathcal{J}(\alpha_{k+1}, \alpha_r, T)$, and $\alpha(T) = \alpha_{k+1}$. Thus, their combined probability is $p_{k+1}^{|\mathcal{R}(T)| - |\mathcal{J}(\alpha_{k+1},\alpha_s, T)| + 1} ((p_{k+1} - p_s)^{|\mathcal{J}(\alpha_{k+1}, \alpha_r, T)|}-p_{k+1}^{|\mathcal{J}(\alpha_{k+1}, \alpha_r, T)|})$. Hence for $k+1$  support, the probability is
\begin{eqnarray*}
& &\sum_{r,s\in[k]:\alpha_r>\alpha_s} p_r^{|\mathcal{R}(T)| - |\mathcal{J}(\alpha_r,\alpha_r, T)| + 1} ((p_r + p_s)^{|\mathcal{J}(\alpha_r,\alpha_s,T)|}-p_r^{|\mathcal{J}(\alpha_r,\alpha_s, T)|}) 
\\
& & \hspace{8 mm} + ~ p_{k+1}^{|\mathcal{R}(T)| - |\mathcal{J}(\alpha_{k+1},\alpha_s, T)| + 1} ((p_{k+1} + p_s)^{|\mathcal{I}(\alpha_{k+1},\alpha_s,T)|}-p_{k+1}^{|\mathcal{J}(\alpha_{k+1}, \alpha_s, T)|}).
\end{eqnarray*}
From case 01 and case 02 with $k+1$ support, we have 
\begin{small}
\begin{eqnarray*}
h_T(p_j, \alpha_j; j \in [k+1]) &=& \sum_{r,s\in[k+1]:\alpha_r>\alpha_s} p_r^{|\mathcal{R}(T)| - |\mathcal{J}(\alpha_r,\alpha_s, T)| + 1} \left((p_r + p_s)^{|\mathcal{J}(\alpha_r,\alpha_s,T)|}-p_r^{|\mathcal{J}(\alpha_r,\alpha_s, T)|}\right)
\\
& & ~~~~~
+ \sum_{a=1}^{k+1} p_a \left(\sum_{b=a}^{k+1}p_b\right)^{|\mathcal{R}(T)|}.
\end{eqnarray*}
\end{small}
Furthermore, it is true for $k+1$ support. Thus, from the principle of Mathematical induction, this is true for any $l \geq 2$.
\end{proof}

\subsection{2 agent 2 support model}
\label{sec: 2_agent_full_info_suppli}
In this Section, we will provide further details about the 2 agents' full information noisy game with 2 support of the noise distribution. First, we consider the following noisy game.
\begin{equation*}
\label{eqn: noisy_game}
\tag{game 1}
\tilde{v}_1(12) > \tilde{v}_1(1); \hspace{2 mm}
\tilde{v}_2(12) > \tilde{v}_2(2).
\end{equation*}
We also consider the other possible noisy games with 2 agents in later subsections. 
\subsubsection{Proof of Lemma \ref{lemma: 2_agent_2_support_full_info}}
\label{proof: 2_agent_2_support}
\begin{proof}
For noisy \ref{eqn: game1}, we have $\tilde{\pi} = N$. Now, consider the noise support $\mathcal{N}_{sp} = \{1, \alpha\}$, where $\alpha > 1$ such that $\mathbb{P}[\alpha(S) = \alpha] = p = 1 - \mathbb{P}[\alpha(S) = 1]$, for some fixed and unknown $p$. Given noisy \ref{eqn: noisy_game}, there are 8 possible combinations of $\alpha$'s (because each coalition has two options). We will now enumerate all such possibilities:
\begin{enumerate}
\item $\alpha(1) =1; \alpha(2) = 1; ~\alpha(12) = 1$. The probability of such alpha is $(1-p)^3$. Thus, the noise-free values are $v_1(1) = \tilde{v}_1(1); ~v_2(2) = \tilde{v}_2(2), v_1(12) = \tilde{v}_1(12) ~ and ~ v_2(12) = \tilde{v}_2(12) $. Therefore, The noise-free game is:
\begin{eqnarray*}
	{v}_1(12) > {v}_1(1); \hspace{2 mm}
	{v}_2(12) > {v}_2(2).
\end{eqnarray*}
From this game we have $\pi = \tilde{\pi} = N$. 

\item $\alpha(1) = 1; \alpha(2) = 1; ~\alpha(12) = \alpha$ Probability of such alpha's is $p(1-p)^2$. Thus the actual values are $v_1(1) = \tilde{v}_1(1); ~v_2(2) = \tilde{v}_2(2), v_1(12) = \frac{\tilde{v}_1(12)}{\alpha} ~ and ~ v_2(12) = \frac{\tilde{v}_2(12)}{\alpha} $.	Therefore, the actual preferences will depend on the relative values of $
\alpha$ and $\tilde{\textbf{\textit{v}}}$. If $\alpha$ and $\tilde{\textbf{\textit{v}}}$'s are such that $\frac{\tilde{v}_1(12)}{\alpha} > \tilde{v}_1(1)$ and $\frac{\tilde{v}_2(12)}{\alpha} > \tilde{v}_2(2)$, then $\pi  = N $, otherwise $\pi = \{\{1\},\{2\}\}$.

\item $\alpha(1) =1; \alpha(2) = \alpha; ~\alpha(12) = 1$. The probability of such alpha is $p(1-p)^2$. Thus, the actual values are $v_1(1) = \tilde{v}_1(1); ~v_2(2) = \frac{\tilde{v}_2(2)}{\alpha}, v_1(12) = \tilde{v}_1(12) ~ and ~ v_2(12) = \tilde{v}_2(12) $. Since, $\tilde{v}_2(12) > \tilde{v}_2(2) > \frac{\tilde{v}_2(2)}{\alpha}$. The noise-free game is:
\begin{eqnarray*}
	{v}_1(12) > {v}_1(1); \hspace{2 mm}
	{v}_2(12) > {v}_2(2).
\end{eqnarray*}
So, we have $\pi = \tilde{\pi}  =  N$.

\item $\alpha(1) =\alpha;~ \alpha(2) = 1; ~\alpha(12) = 1$. Probability of such alpha's is $p(1-p)^2$. Thus, the actual values are $v_1(1) = \frac{\tilde{v}_1(1)}{\alpha};~ ~v_2(2) = \tilde{v}_2(2), v_1(12) = \tilde{v}_1(12)~ and ~ v_2(12) = \tilde{v}_2(12)$. Since, $\tilde{v}_1(12) > \tilde{v}_1(1) > \frac{\tilde{v}_1(1)}{\alpha}.$ Therefore The noise-free game is:
\begin{eqnarray*}
	{v}_1(12) > {v}_1(1); \hspace{2 mm}
	{v}_2(12) > {v}_2(2).
\end{eqnarray*}
From this game we have $\pi = \tilde{\pi} = N$.

\item $\alpha(1) =1; \alpha(2) = \alpha; ~\alpha(12) = \alpha$. The probability of this alpha is $p^2(1-p)$. Thus, the actual values are $v_1(1) = \tilde{v}_1(1); ~v_2(2) = \frac{\tilde{v}_2(2)}{\alpha}, v_1(12) = \frac{\tilde{v}_1(12)}{\alpha} ~ and ~ v_2(12) = \frac{\tilde{v}_2(12)}{\alpha} $. The actual preferences will depend on the relative values of $\alpha$ and $\tilde{\textbf{\textit{v}}}$. If $\alpha$ and $\tilde{\textbf{\textit{v}}}$'s are such that $\frac{\tilde{v}_1(12)}{\alpha} > \tilde{v}_1(1)$, then $\pi  = N $, otherwise $\pi = \{\{1\},\{2\}\}$.

\item $\alpha(1) =\alpha; \alpha(2) = 1; ~\alpha(12) = \alpha$. The probability of such alpha is $p^2(1-p)$. Thus, the actual values are $v_1(1) = \frac{\tilde{v}_1(1)}{\alpha}, ~~ v_2(2) = \tilde{v}_2(2); ~v_1(12) = \frac{\tilde{v}_1(12)}{\alpha} ~ and ~ v_2(12) = \frac{\tilde{v}_2(12)}{\alpha} $. The actual preferences will depend on the relative values of $\alpha$ and $\tilde{\textbf{\textit{v}}}$. If $\alpha$ and $\tilde{\textbf{\textit{v}}}$'s are such that  $\frac{\tilde{v}_2(12)}{\alpha} > \tilde{v}_2(2)$, then $\pi  = N $ otherwise $\pi = \{\{1\},\{2\}\}$.

\item $\alpha(1) =\alpha; \alpha(2) = \alpha; ~\alpha(12) = 1$. Probability of such alpha's is $p^2(1-p)$. Thus, the actual values are $v_1(1) = \frac{\tilde{v}_1(1)}{\alpha}, ~~ v_2(2) = \frac{\tilde{v}_2(2)}{\alpha}; ~v_1(12) = \tilde{v}_1(12) ~ and ~ v_2(12) = \tilde{v}_2(12)$. Since, $\tilde{v}_1(12) > \tilde{v}_1(1) > \frac{\tilde{v}_1(1)}{\alpha}.$  and, $\tilde{v}_2(12) > \tilde{v}_2(2) > \frac{\tilde{v}_2(2)}{\alpha}.$ The noise-free game is: 
\begin{eqnarray*}
	{v}_1(12) > {v}_1(1); \hspace{2 mm}
	{v}_2(12) > {v}_2(2).
\end{eqnarray*}
From this game we have $\pi = \tilde{\pi} = N$.

\item $\alpha(1) =\alpha; \alpha(2) = \alpha; ~\alpha(12) = \alpha$. The probability of such alpha is $p^3$. Thus, the actual values are $v_1(1) = \frac{\tilde{v}_1(1)}{\alpha}, ~~ v_2(2) = \frac{\tilde{v}_2(2)}{\alpha}; ~v_1(12) = \frac{\tilde{v}_1(12)}{\alpha} ~ and ~ v_2(12) = \frac{\tilde{v}_2(12)}{\alpha}$. Therefore, the noise-free game is:
\begin{eqnarray*}
	{v}_1(12) > {v}_1(1); \hspace{2 mm}
	{v}_2(12) > {v}_2(2).
\end{eqnarray*}
From this game it is clear that $\pi = \tilde{\pi} = N$.
\end{enumerate}
Recall, $\overline{r}= \max \left\lbrace \frac{\tilde{v}_1(12)}{\tilde{v}_1(1)}, \frac{\tilde{v}_2(12)}{\tilde{v}_2(2)} \right\rbrace$, and $\underline{r}=\min \left\lbrace \frac{\tilde{v}_1(12)}{\tilde{v}_1(1)}, \frac{\tilde{v}_2(12)}{\tilde{v}_2(2)} \right\rbrace$. Out of 8 cases there are 5 cases (case 1,3,4,7,8) in which the grand coalition $\pi= \tilde{\pi} = N$ is formed in noise-free game. In these conditions, the relative value of $\tilde{v}_1(\cdot), \tilde{v}_2(\cdot)$ should satisfy $\alpha \geq \overline{r}$, and this constitute the first expression $p^3 + p^2(1-p) +  2p(1-p)^2 + (1-p)^3$ of $\mathbb{P}[\pi =\tilde{\pi}~|~game~1]$. 
Apart from this, if the inflation interval is $\underline{r} \leq \alpha < \overline{r}$, then $\pi = \tilde{\pi} = N$ is also possible from case (6) with probability $p^2(1-p)$. Thus,  $p^2(1-p)$ will be added to the above prediction probability. So, we have $\mathbb{P}[\pi =\tilde{\pi}~|~game~1]$ corresponding to it. Moreover, finally, if $\alpha <\underline{r}$, all cases are allowable, and hence the grand coalition will always form in the noise-free game. Thus,
\begin{equation}
\mathbb{P}[\pi =\tilde{\pi}~|~game~1]= \begin{cases}
p^3 + p^2(1-p) +  2p(1-p)^2 + (1-p)^3, & if ~\alpha \geq \overline{r}
\\
p^3 + 2p^2(1-p) +  2p(1-p)^2 + (1-p)^3, & if~ \underline{r} \leq \alpha < \overline{r}  
\\
1, & if~ \alpha < \underline{r}.
\end{cases}  
\end{equation}
Simplifying these polynomials, we have
\begin{equation}
\begin{aligned}
\mathbb{P}[\pi =\tilde{\pi}~|~game~1] = \begin{cases}
1-p(1-p^2), & if ~\alpha \geq \overline{r}
\\
1-p(1-p), & if~ \underline{r} \leq \alpha < \overline{r}
\\
1, & if~ \alpha < \underline{r}.
\end{cases}
\end{aligned}
\end{equation}
This ends the proof.
\end{proof}
If we allow some user given satisfaction $\zeta$ on the prediction probability, i.e., $\mathbb{P}[\pi =\tilde{\pi}~|~game~1] = \zeta$, we get the following noise interval
\begin{equation}
\label{eqn: zeta=0.9_game_1}
    I^{\star}(\zeta = 0.9)  = \begin{cases} [0,0.101] \cup [0.946,1], & ~if~ \alpha \geq \overline{r};
    \\
    [0,0.113]\cup [0.887,1], & ~if~ \underline{r} \leq \alpha < \overline{r}
    \\
    1, & ~if ~ \alpha < \underline{r}.
    \end{cases}
\end{equation}

\subsubsection{Details of the other 2 agent noisy games}
\label{subsec: other_noisy_game_2_agent}
Here we will give the prediction probabilities for other possible noisy games with 2 agents and 2 noise support.

\noindent \textbf{Both agents prefer staying alone in noisy game}

As opposed to the noisy \ref{eqn: game1}, in noisy \ref{eqn: game_2} both agents prefer to stay alone. The noisy preferences of agents are as follows:
\begin{equation}
\label{eqn: game_2}
\tag{game 2}
\tilde{v}_1(1) > \tilde{v}_1(12); \hspace{2 mm} \tilde{v}_2(2) > \tilde{v}_2(12).
\end{equation}
Clearly $\tilde{\pi} = \{\{1\},\{2\}\} \neq N$ is the core-stable outcome. The following lemma provides prediction probability, $\mathbb{P}[\pi = \tilde{\pi} ~|~game~2]$ for noisy \ref{eqn: game_2}.

\begin{lemma}
\label{lemma: singletons_prefered_both}
For noisy \ref{eqn: game_2} with full information of $\tilde{\textbf{\textit{v}}}$'s, the prediction probability that \textit{unknown} noise-free game has $\pi = \tilde{\pi}$ as a core-stable outcome is
\begin{equation}
\label{eqn: both_alone_eta}
\begin{aligned}
\mathbb{P}[\pi =\tilde{\pi}~|~ game~2] = \begin{cases}
1-p^2(1-p), & if ~\frac{1}{\alpha} < \underline{r}	
\\
1, &  if ~\frac{1}{\alpha} \geq \underline{r}.
\end{cases}
\end{aligned}
\end{equation}
Moreover, the minimal and maximal values of above prediction probability are $0.85$ (when $p = 2/3$), and $1$, respectively.
\end{lemma}
Similar to \ref{eqn: game1}, the probability of formation of partition $\pi = \{\{1\},\{2\}\}$ in an \textit{unknown} noise-free game is always more than $0.85$. So, the safety value is $0.85$. The prediction probability is $1$ when $\frac{1}{\alpha} \geq \underline{r}$ for any noise probability $p$. Moreover, for some user-given satisfaction $\zeta$, we obtain the corresponding $p$ by setting $\mathbb{P}[\pi =\{\{1\},\{2\}\}~|~ game~2] = \zeta$. In particular, we have 
\begin{equation}
\label{eqn: zeta=0.9_game_2}
I^{\star}(\zeta = 0.9) = \begin{cases} [0,0.413] \cup [0.867, 1], & ~if ~\frac{1}{\alpha} < \underline{r}
\\
[0,1], & ~\frac{1}{\alpha} \geq \underline{r}.
\end{cases}
\end{equation}
It is easy to see that the allowable $p$ is  larger than the interval given in Equation \eqref{eqn: zeta=0.9_game_1} for \ref{eqn: game1}. So, the partition $\tilde{\pi} = \{\{1\},\{2\}\}$ is noise robust for larger number of inflation probabilities $p$. Again the noise set will shrink if we increase the satisfaction $\zeta$.
\begin{figure}[h!]
\centering
\includegraphics[scale = 0.7]{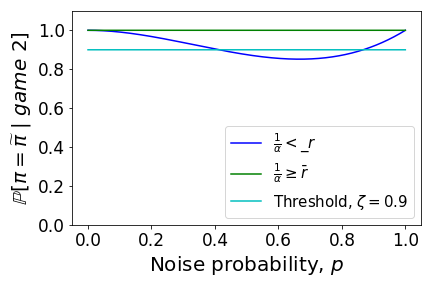}
\caption{The prediction probability $\mathbb{P}[\pi = \tilde{\pi} ~|~ game~2]$. For $\zeta = 0.9$, the  noise regimes are given in Equation \eqref{eqn: zeta=0.9_game_2}.}
\label{fig: game_2}
\end{figure}

\noindent \textbf{Agent 1 prefers to stay alone and agent 2 prefers grand coalition in noisy game}

\textcolor{black}{Now, we consider a noisy game where agent 1 prefers to stay alone, whereas agent 2 prefers the grand coalition. In particular, the preferences in the noisy game are
\begin{equation}
\label{eqn: game_3}
\tag{game 3}
\tilde{v}_1(1) > \tilde{v}_1(12); \hspace{2 mm}   \tilde{v}_2(12) > \tilde{v}_2(2).
\end{equation}
Again $\tilde{\pi} = \{\{1\},\{2\}\} \neq N$ is noisy core-stable outcome. The prediction probability, $\mathbb{P}[\pi =\tilde{\pi}~|~ game~3] $ is given in the Lemma below.}
\begin{lemma}
\label{lemma: singelton_by_one}
For noisy \ref{eqn: game_3} with full information of $\tilde{\textbf{\textit{v}}}$'s, the prediction probability that unknown noise-free game has $\pi = \tilde{\pi}$ as a core-stable outcome is given by:
\begin{equation}
\label{eqn: agent_1_alone_eta}
\begin{aligned}
\mathbb{P}[\pi =\tilde{\pi}~|~ game~3] = \begin{cases}
1-p(1-p), & if ~ \frac{1}{\alpha} < \frac{\tilde{v}_1(12)}{\tilde{v}_1(1)} 
\\
1, & if ~ \frac{1}{\alpha} \geq \frac{\tilde{v}_1(12)}{\tilde{v}_1(1)}. 
\end{cases}
\end{aligned}
\end{equation}
Moreover, the minimal and maximal values of above prediction probability are $0.75$ (when $p = 0.5$), and $1$, respectively.
\end{lemma}

\textcolor{black}{Similar to \ref{eqn: game1} and \ref{eqn: game_2} the probability of formation of partition $\pi = \{\{1\},\{2\}\}$ in an \textit{unknown} noise-free game is always more than $0.75$ that is the safety value for \ref{eqn: game_3}. The prediction probability is $1$ when $\frac{1}{\alpha} \geq \frac{\tilde{v}_1(12)}{\tilde{v}_1(1)}$ for any noise probability $p$. Moreover, for some user-given satisfaction, $\zeta$ we obtain the corresponding $p$ by setting $\mathbb{P}[\pi =\{\{1\},\{2\}\}~|~ game~3] = \zeta$. In particular, 
\begin{equation}
\label{eqn: zeta=0.9_game_3}
I^{\star}(\zeta = 0.9) = \begin{cases} [0,0.113] \cup [0.887, 1], & ~if ~ \frac{1}{\alpha} < \frac{\tilde{v}_1(12)}{\tilde{v}_1(1)} 
\\
[0,1], & ~if~ \frac{1}{\alpha} \geq \frac{\tilde{v}_1(12)}{\tilde{v}_1(1)}.
\end{cases}
\end{equation}}
The following figure shows the prediction probabilities for \ref{eqn: game_3}.
\begin{figure}[h!]
\centering
\centering
\includegraphics[scale = 0.7]{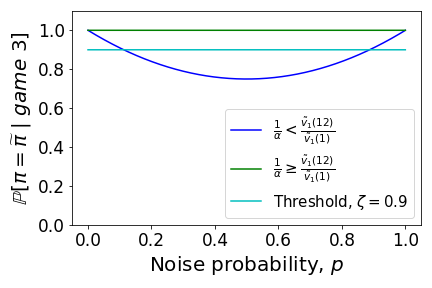}
\caption{The prediction probability $\mathbb{P}[\tilde{\pi} = \pi ~|~ game~3]$. For $\zeta = 0.9$, we obtain the noise regimes as given in Equation \eqref{eqn: zeta=0.9_game_3}.}
\label{fig: game_3}
\end{figure}

\noindent \textbf{Agent 1 prefers grand coalition and agent 2 prefers to stay alone}

Finally, consider a noisy game symmetric to \ref{eqn: game_3}. Here agent 1 prefers a grand coalition, and agent 2 prefers to stay alone. In particular, we have the following preferences.
\begin{equation}
\label{eqn: game_4}
\tag{game 4}
\tilde{v}_1(12) > \tilde{v}_1(1); \hspace{2mm} \tilde{v}_2(2) > \tilde{v}_2(12).
\end{equation}
Again $\tilde{\pi} = \{\{1\},\{2\}\} \neq 
N$ is a noisy core-stable outcome. In the following lemma, we find the prediction probability when noisy \ref{eqn: game_4} is considered.
\begin{lemma}
\label{lemma: singelton_by_two}
For noisy \ref{eqn: game_4} with full information of $\tilde{\textbf{\textit{v}}}$'s, the prediction probability that noise-free game has $\pi = \tilde{\pi}$ as as core-stable outcome is given by:
\begin{equation}
\label{eqn: agent_2_alone_eta}
\begin{aligned}
\mathbb{P}[\pi =\tilde{\pi}~|~ game~4] = \begin{cases}
1-p(1-p), & if ~ \frac{1}{\alpha} < \frac{\tilde{v}_2(12)}{\tilde{v}_2(2)} 
\\
1, & if ~ \frac{1}{\alpha} \geq \frac{\tilde{v}_2(12)}{\tilde{v}_2(2)}.
\end{cases}
\end{aligned}
\end{equation}
So, the minimal and maximal values of  above prediction probability are $0.75$ (when $p = 0.5$) and  $1$ respectively.
\end{lemma}
In this case also, the noise regime can be obtained using $\mathbb{P}[\pi =\tilde{\pi}~|~ game~4] = \zeta$. In particular,
\begin{equation}
\label{eqn: zeta=0.9_game_4}
I^{\star}(\zeta = 0.9) = \begin{cases} [0,0.113] \cup [0.887, 1], & ~if ~ \frac{1}{\alpha} < \frac{\tilde{v}_2(12)}{\tilde{v}_2(2)} 
\\
[0,1], & ~if~ \frac{1}{\alpha} \geq \frac{\tilde{v}_2(12)}{\tilde{v}_2(2)}.
\end{cases}
\end{equation}
Figure \ref{fig: game_4} shows the prediction probabilities for \ref{eqn: game_4}.
\begin{figure}[h!]
\centering
\includegraphics[scale = 0.7]{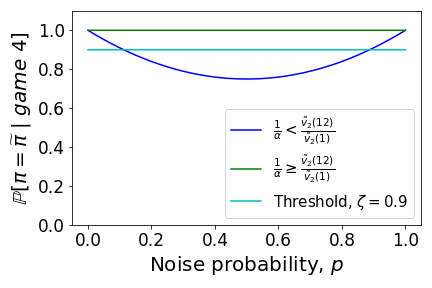}
\caption{The prediction probability $\mathbb{P}[\pi = \tilde{\pi}~|~ game~4]$. For $\zeta = 0.9$, we obtain the noise regimes as given in Equation \eqref{eqn: zeta=0.9_game_4}.}
\label{fig: game_4} 
\end{figure} 

\subsection{2 agents 3 support noise model}
\label{app: 3_supprt_full_info}
In this section, we consider two player noisy hedonic game with three support noise model, i.e., $\mathcal{N}_{sp} = \{1,\alpha_1, \alpha_2\}$, with $\alpha_1 > 1$, and $\alpha_2 < 1$. Note that $\alpha_1, \alpha_2 > 0$. Let $\mathbb{P}[\alpha(S) = \alpha_1] = p_1; ~ \mathbb{P}[\alpha(S) = \alpha_2] = p_2; ~ and~ \mathbb{P}[\alpha(S) = 1] = 1-p_1-p_2$. That is the value of each coalition is either inflated with probability $p_1$, or deflated with probability $p_2$ or retained with probability $1-p_1-p_2$. The following lemma provides the prediction probability for \ref{eqn: game1}. 

\subsubsection{Proof of Lemma \ref{lemma: 2_agent_3_support}}
\label{sec: proof_3_suppot_lemma}

\begin{proof}
For \ref{eqn: game1}, with $l=3$ support of noise there are 27 possible cases for $\alpha$'s. Since there are 3 coalitions, each coalition's value can either be retained, inflated by $\alpha_1$, or deflated by $\alpha_2$. We will now enumerate all of them: 
\begin{enumerate}
	\item $\alpha(1) =1; \alpha(2) = 1; ~\alpha(12) = 1$ Probability of such alpha's is $(1-p_1-p_2)^3$. Thus, the actual values are $v_1(1) = \tilde{v}_1(1); ~v_2(2) = \tilde{v}_2(2), v_1(12) = \tilde{v}_1(12) ~ and ~ v_2(12) = \tilde{v}_2(12) $. The noise-free game is: ${v}_1(12) > {v}_1(1)$; \hspace{1 mm} ${v}_2(12) > {v}_2(2)$.
	So, $\pi = \tilde{\pi}$ in this case.
	
	\item $\alpha(1) = 1; \alpha(2) = 1; ~\alpha(12) = \alpha_1$. Probability of such alpha's is $p_1(1-p_1-p_2)^2$. Thus, the actual values are $v_1(1) = \tilde{v}_1(1); ~v_2(2) = \tilde{v}_2(2), v_1(12) = \frac{\tilde{v}_1(12)}{\alpha_1} ~ and ~ v_2(12) = \frac{\tilde{v}_2(12)}{\alpha_1} $. The noise-free game preferences are unclear; they will depend on the relative values of $\alpha_1$ and $\tilde{v}$. If $\alpha_1$ and $\tilde{v}$'s are such that $\frac{\tilde{v}_1(12)}{\alpha_1} > \tilde{v}_1(1)$ and $\frac{\tilde{v}_2(12)}{\alpha_1} > \tilde{v}_2(2)$ then $\pi  = N $, otherwise $\pi = \{\{1\},\{2\}\}$.
	
	\item $\alpha(1) =1; \alpha(2) = \alpha_1; ~\alpha(12) = 1$. Probability of such alpha's is $p_1(1-p_1-p_2)^2$. Thus, the actual values are $v_1(1) = \tilde{v}_1(1); ~v_2(2) = \frac{\tilde{v}_2(2)}{\alpha_1}, v_1(12) = \tilde{v}_1(12) ~ and ~ v_2(12) = \tilde{v}_2(12) $. Since $\tilde{v}_2(12) > \tilde{v}_2(2) > \frac{\tilde{v}_2(2)}{\alpha_1}.$ The noise-free game is: ${v}_1(12) > {v}_1(1)$; \hspace{1 mm} ${v}_2(12) > {v}_2(2)$.
	So, $\pi = \tilde{\pi}$ in this case.
	
	\item $\alpha(1) =\alpha_1;~ \alpha(2) = 1; ~\alpha(12) = 1$. Probability of such alpha's is $p_1(1-p_1-p_2)^2$. Thus, the actual values are $v_1(1) = \frac{\tilde{v}_1(1)}{\alpha_1};~ ~v_2(2) = \tilde{v}_2(2), v_1(12) = \tilde{v}_1(12)~ and ~ v_2(12) = \tilde{v}_2(12)$. Since, $\tilde{v}_1(12) > \tilde{v}_1(1) > \frac{\tilde{v}_1(1)}{\alpha_1}$. The noise-free game is: ${v}_1(12) > {v}_1(1)$; \hspace{1 mm} ${v}_2(12) > {v}_2(2)$.
	So, $\pi = \tilde{\pi}$ in this case.
	
	\item $\alpha(1) =1; \alpha(2) = \alpha_1; ~\alpha(12) = \alpha_1$. Probability of such alpha's is $p_1^2(1-p_1-p_2)$. Thus, the actual values are $v_1(1) = \tilde{v}_1(1); ~v_2(2) = \frac{\tilde{v}_2(2)}{\alpha_1}, v_1(12) = \frac{\tilde{v}_1(12)}{\alpha_1} ~ and ~ v_2(12) = \frac{\tilde{v}_2(12)}{\alpha_1} $. The noise-free game preferences are unclear; they will depend on the relative values of $\alpha_1$ and $\tilde{v}$. If $\alpha_1$ and $\tilde{v}$'s are such that $\frac{\tilde{v}_1(12)}{\alpha_1} > \tilde{v}_1(1)$, then $\pi  = N $, otherwise $\pi = \{\{1\},\{2\}\}$.
	
	\item $\alpha(1) =\alpha_1; \alpha(2) = 1; ~\alpha(12) = \alpha_1$. Probability of such alpha's is $p_1^2(1-p_1-p_2)$. Thus, the  actual values are $v_1(1) = \frac{\tilde{v}_1(1)}{\alpha_1}, ~~ v_2(2) = \tilde{v}_2(2); ~v_1(12) = \frac{\tilde{v}_1(12)}{\alpha_1} ~ and ~ v_2(12) = \frac{\tilde{v}_2(12)}{\alpha_1} $. The noise-free game preferences are unclear; they will depend on the relative values of $\alpha_1$ and $\tilde{v}$. If $\alpha_1$ and $\tilde{v}$'s are such that  $\frac{\tilde{v}_2(12)}{\alpha_1} > \tilde{v}_2(2)$, then $\pi  = N $, otherwise $\pi = \{\{1\},\{2\}\}$.
	
	\item $\alpha(1) =\alpha_1; \alpha(2) = \alpha_1; ~\alpha(12) = 1$. Probability of such alpha's is $p_1^2(1-p_1-p_2)$. Thus, the actual values are $v_1(1) = \frac{\tilde{v}_1(1)}{\alpha_1}, ~~ v_2(2) = \frac{\tilde{v}_2(2)}{\alpha_1}; ~v_1(12) = \tilde{v}_1(12) ~ and ~ v_2(12) = \tilde{v}_2(12)$. Since, $\tilde{v}_1(12) > \tilde{v}_1(1) > \frac{\tilde{v}_1(1)}{\alpha_1}.$  and, $\tilde{v}_2(12) > \tilde{v}_2(2) > \frac{\tilde{v}_2(2)}{\alpha_1}.$  The noise-free game is: ${v}_1(12) > {v}_1(1)$; \hspace{1 mm} ${v}_2(12) > {v}_2(2)$.
	So, $\pi = \tilde{\pi}$ in this case.
	
	\item $\alpha(1) =\alpha_1; \alpha(2) = \alpha_1; ~\alpha(12) = \alpha_1$. The probability of such alpha is $p_1^3$. Thus, the actual values are $v_1(1) = \frac{\tilde{v}_1(1)}{\alpha_1}, ~~ v_2(2) = \frac{\tilde{v}_2(2)}{\alpha_1}; ~v_1(12) = \frac{\tilde{v}_1(12)}{\alpha_1} ~ and ~ v_2(12) = \frac{\tilde{v}_2(12)}{\alpha_1}$. The noise-free game is: ${v}_1(12) > {v}_1(1)$; \hspace{1 mm} ${v}_2(12) > {v}_2(2)$.
	So, $\pi = \tilde{\pi}$ in this case.

	\item $\alpha(1) = 1; \alpha(2) = 1; ~\alpha(12) = \alpha_2$. Probability of such alpha's is $p_2(1-p_1-p_2)^2$. Thus, the actual values are $v_1(1) = \tilde{v}_1(1); ~v_2(2) = \tilde{v}_2(2), v_1(12) = \frac{\tilde{v}_1(12)}{\alpha_2} ~ and ~ v_2(12) = \frac{\tilde{v}_2(12)}{\alpha_2} $. Since $\alpha_2<1$, thus $\frac{\tilde{v}_1(12)}{\alpha_2} > \tilde{v}_1(12) > \tilde{v}_1(1)  = v_1(1)$. Similarly, $\frac{\tilde{v}_2(12)}{\alpha_2} > \tilde{v}_2(12) > \tilde{v}_2(2)  = v_2(2)$. The noise-free game is: ${v}_1(12) > {v}_1(1)$; \hspace{1 mm} ${v}_2(12) > {v}_2(2)$.
	So, $\pi = \tilde{\pi}$ in this case.
	
	\item $\alpha(1) =1; \alpha(2) = \alpha_2; ~\alpha(12) = 1$. Probability of these alpha's is $p_2(1-p_1-p_2)^2$. Thus, the actual values are $v_1(1) = \tilde{v}_1(1); ~v_2(2) = \frac{\tilde{v}_2(2)}{\alpha_2}, v_1(12) = \tilde{v}_1(12) ~ and ~ v_2(12) = \tilde{v}_2(12) $. The noise-free game preferences are unclear; they will depend on the relative values of $\alpha_2$ and $\tilde{v}$. If $\alpha_2$ and $\tilde{v}$'s are such that  $\tilde{v}_2(12) > \frac{\tilde{v}_2(2)}{\alpha_2}$  then $\pi  = N $, otherwise $\pi = \{\{1\},\{2\}\}$.
	
	\item $\alpha(1) =\alpha_2;~ \alpha(2) = 1; ~\alpha(12) = 1$. Probability of such alpha's is $p_2(1-p_1-p_2)^2$. Thus, the actual values are $v_1(1) = \frac{\tilde{v}_1(1)}{\alpha_2};~ ~v_2(2) = \tilde{v}_2(2), v_1(12) = \tilde{v}_1(12)~ and ~ v_2(12) = \tilde{v}_2(12)$. The noise-free game preferences are unclear; they will depend on the relative values of $\alpha_2$ and $\tilde{v}$. If $\alpha_2$ and $\tilde{v}$'s are such that  $\tilde{v}_1(12) > \frac{\tilde{v}_1(1)}{\alpha_2}$ then $\pi  = N $, otherwise $\pi = \{\{1\},\{2\}\}$.
	
	\item $\alpha(1) =1; \alpha(2) = \alpha_2; ~\alpha(12) = \alpha_2$. probability of such alpha's is $p_2^2(1-p_1-p_2)$. Thus, the actual values are $v_1(1) = \tilde{v}_1(1); ~v_2(2) = \frac{\tilde{v}_2(2)}{\alpha_2}, v_1(12) = \frac{\tilde{v}_1(12)}{\alpha_2} ~ and ~ v_2(12) = \frac{\tilde{v}_2(12)}{\alpha_2} $
	Since $\alpha_2<1$, thus $\frac{\tilde{v}_1(12)}{\alpha_2} > \tilde{v}_1(12) > \tilde{v}_1(1)  = v_1(1)$, and $\frac{\tilde{v}_2(12)}{\alpha_2} > \frac{\tilde{v}_2(2)}{\alpha_2}$. The noise-free game is: ${v}_1(12) > {v}_1(1)$; \hspace{1 mm} ${v}_2(12) > {v}_2(2)$.
	So, $\pi = \tilde{\pi}$ in this case.
	
	\item $\alpha(1) =\alpha_2; \alpha(2) = 1; ~\alpha(12) = \alpha_2$. Probability of such alpha's is $p_2^2(1-p_1-p_2)$. Thus, the actual values are $v_1(1) = \frac{\tilde{v}_1(1)}{\alpha_2}, ~~ v_2(2) = \tilde{v}_2(2); ~v_1(12) = \frac{\tilde{v}_1(12)}{\alpha_2} ~ and ~ v_2(12) = \frac{\tilde{v}_2(12)}{\alpha_2} $. Since $\alpha_2<1$ thus $\frac{\tilde{v}_2(12)}{\alpha_2} > \tilde{v}_2(12) > \tilde{v}_2(2)  = v_2(2)$, and $\frac{\tilde{v}_1(12)}{\alpha_2} > \frac{\tilde{v}_1(1)}{\alpha_2}$. The noise-free game is: ${v}_1(12) > {v}_1(1)$; \hspace{1 mm} ${v}_2(12) > {v}_2(2)$.
So, $\pi = \tilde{\pi}$ in this case.
	
	\item $\alpha(1) =\alpha_2; \alpha(2) = \alpha_2; ~\alpha(12) = 1$. Probability of such alpha's is $p_2^2(1-p_1-p_2)$. Thus, the actual values are $v_1(1) = \frac{\tilde{v}_1(1)}{\alpha_2}, ~~ v_2(2) = \frac{\tilde{v}_2(2)}{\alpha_2}; ~v_1(12) = \tilde{v}_1(12) ~ and ~ v_2(12) = \tilde{v}_2(12)$. The noise-free game preferences are unclear; they will depend on the relative values of $\alpha_2$ and $\tilde{v}$. If $\alpha_2$ and $\tilde{v}$'s are such that $\tilde{v}_1(12) > \frac{\tilde{v}_1(1)}{\alpha_2}   $ and $\tilde{v}_1(12) > \frac{\tilde{v}_2(2)}{\alpha_2} $ then $\pi  = N $, otherwise $\pi = \{\{1\},\{2\}\}$.
	
	\item $\alpha(1) =1; \alpha(2) = \alpha_1; ~\alpha(12) = \alpha_2$. Probability of such alpha's is $ p_1 p_2 (1-p_1-p_2)$. Thus, the actual values are $v_1(1) = \tilde{v}_1(1), ~~ v_2(2) = \frac{\tilde{v}_2(2)}{\alpha_1}; ~v_1(12) = \frac{\tilde{v}_1(12)}{\alpha_2} ~ and ~ v_2(12) = \frac{\tilde{v}_2(12)}{\alpha_2}$. Since $\frac{\tilde{v}_1(12)}{\alpha_2} > \tilde{v}_1(12) > \tilde{v}_1(1)$ and $\frac{\tilde{v}_2(12)}{\alpha_2} > \tilde{v}_2(12) > \tilde{v}_2(2)  > \frac{\tilde{v}_2(2)}{\alpha_1} $. The noise-free game is: ${v}_1(12) > {v}_1(1)$; \hspace{1 mm} ${v}_2(12) > {v}_2(2)$.
	So, $\pi = \tilde{\pi}$ in this case.
	
	\item $\alpha(1) =1; \alpha(2) = \alpha_2; ~\alpha(12) = \alpha_1$. Probability of such alpha's is $p_1 p_2 (1-p_1-p_2)$. Thus, actual values are $v_1(1) = \tilde{v}_1(1), ~~ v_2(2) = \frac{\tilde{v}_2(2)}{\alpha_2}; ~v_1(12) = \frac{\tilde{v}_1(12)}{\alpha_1} ~ and ~ v_2(12) = \frac{\tilde{v}_2(12)}{\alpha_1}$. The noise-free game preferences are unclear; it will depend on the relative values $\alpha_1$, $\alpha_2$ and $\tilde{v}$. If $\alpha_1$, $\alpha_2$ and $\tilde{v}$'s are such that $\frac{\tilde{v}_1(12)}{\alpha_1} > \tilde{v}_1(1) $ and $\frac{\tilde{v}_2(12)}{\alpha_1} > \frac{\tilde{v}_2(2)}{\alpha_2}$ then $\pi  = N $, otherwise $\pi = \{\{1\},\{2\}\}$.
	
	\item $\alpha(1) =\alpha_1; \alpha(2) = 1; ~\alpha(12) = \alpha_2$. Probability of such alpha's is $p_1 p_2 (1-p_1-p_2)$. Thus, the actual values are $v_1(1) = \frac{\tilde{v}_1(1)}{\alpha_1}, ~~ v_2(2) = \tilde{v}_2(2); ~v_1(12) = \frac{\tilde{v}_1(12)}{\alpha_2} ~ and ~ v_2(12) = \frac{\tilde{v}_2(12)}{\alpha_2}$. Since $\frac{\tilde{v}_1(12)}{\alpha_2} > \tilde{v}_1(12) > \tilde{v}_1(1)$ and $\frac{\tilde{v}_2(12)}{\alpha_2} > \tilde{v}_2(12) > \tilde{v}_2(2)  $. The noise-free game is: ${v}_1(12) > {v}_1(1)$; \hspace{1 mm} ${v}_2(12) > {v}_2(2)$.
So, $\pi = \tilde{\pi}$ in this case.
	
	\item $\alpha(1) =\alpha_2; \alpha(2) = 1; ~\alpha(12) = \alpha_1$. Probability of such alpha's is $p_1 p_2 (1-p_1-p_2)$. Thus, the actual values are $v_1(1) = \frac{\tilde{v}_1(1)}{\alpha_2}, ~~ v_2(2) = \tilde{v}_2(2); ~v_1(12) = \frac{\tilde{v}_1(12)}{\alpha_1} ~ and ~ v_2(12) = \frac{\tilde{v}_2(12)}{\alpha_1}$. The noise-free game preferences are unclear; it will depend on the relative values of $\alpha_1$, $\alpha_2$, and $\tilde{v}$. If $\alpha_1$, $\alpha_2$ and $\tilde{v}$'s are such that $\frac{\tilde{v}_1(12)}{\alpha_1} > \frac{\tilde{v}_1(1)}{\alpha_2} $ and $\frac{\tilde{v}_2(12)}{\alpha_1} > \tilde{v}_2(2)$ then $\pi  = N $, otherwise $\pi = \{\{1\},\{2\}\}$.

	\item $\alpha(1) =\alpha_1; \alpha(2) = \alpha_2; ~\alpha(12) = 1$. Probability of such alpha's is $p_1 p_2 (1-p_1-p_2)$. Thus, the actual values are $v_1(1) = \frac{\tilde{v}_1(1)}{\alpha_1}, ~~ v_2(2) = \frac{\tilde{v}_2(2)}{\alpha_2}; ~v_1(12) = \tilde{v}_1(12)~ and ~ v_2(12) = \tilde{v}_2(12)$. The noise-free game preferences are unclear; it will depend on the relative values of $\alpha_1$, $\alpha_2$, and $\tilde{v}$. If $\alpha_1$, $\alpha_2$ and $\tilde{v}$'s are such that $\tilde{v}_2(12) > \frac{\tilde{v}_2(2)}{\alpha_2} $ then $\pi  = N $ otherwise $\pi = \{\{1\},\{2\}\}$.

	\item $\alpha(1) =\alpha_2; \alpha(2) = \alpha_1; ~\alpha(12) = 1$. Probability of such alpha's is $p_1 p_2 (1-p_1-p_2)$. Thus, the actual values are $v_1(1) = \frac{\tilde{v}_1(1)}{\alpha_2}, ~~ v_2(2) = \frac{\tilde{v}_2(2)}{\alpha_1}; ~v_1(12) = \tilde{v}_1(12)~ and ~ v_2(12) = \tilde{v}_2(12)$. The noise-free game preferences are unclear; it will depend on the relative values of $\alpha_1$, $\alpha_2$, and $\tilde{v}$. If $\alpha_1$, $\alpha_2$ and $\tilde{v}$'s are such that $\tilde{v}_1(12) > \frac{\tilde{v}_1(1)}{\alpha_2}$ then $\pi  = N $, otherwise $\pi = \{\{1\},\{2\}\}$
	
	\item $\alpha(1) =\alpha_1; \alpha(2) = \alpha_1; ~\alpha(12) = \alpha_2$. Probability of such alpha's is $p_1^2 p_2$. Thus, the actual values are $v_1(1) = \frac{\tilde{v}_1(1)}{\alpha_1}, ~~ v_2(2) = \frac{\tilde{v}_2(2)}{\alpha_1}; ~v_1(12) = \frac{\tilde{v}_1(12)}{\alpha_2}~ and ~ v_2(12) = \frac{\tilde{v}_2(12)}{\alpha_2}$. Since, $\frac{\tilde{v}_1(12)}{\alpha_2} > \tilde{v}_1(12) > \tilde{v}_1(1) > \frac{\tilde{v}_1(1)}{\alpha_1}$, and $\frac{\tilde{v}_2(12)}{\alpha_2} > \tilde{v}_2(12) > \tilde{v}_2(2) > \frac{\tilde{v}_2(2)}{\alpha_1}$. The noise-free game is: ${v}_1(12) > {v}_1(1)$; \hspace{1 mm} ${v}_2(12) > {v}_2(2)$.
	So, $\pi = \tilde{\pi}$ in this case.
	
	\item $\alpha(1) =\alpha_1; \alpha(2) = \alpha_2; ~\alpha(12) = \alpha_1$. Probability of such alpha's is $p_1^2 p_2$. Thus, the actual values are $v_1(1) = \frac{\tilde{v}_1(1)}{\alpha_1}, ~~ v_2(2) = \frac{\tilde{v}_2(2)}{\alpha_2}; ~v_1(12) = \frac{\tilde{v}_1(12)}{\alpha_1}~ and ~ v_2(12) = \frac{\tilde{v}_2(12)}{\alpha_1}$. The noise-free game preferences are unclear; it will depend on the relative values of $\alpha_1$, $\alpha_2$, and $\tilde{v}$. If $\alpha_1$, $\alpha_2$ and $\tilde{v}$'s are such that $\frac{\tilde{v}_2(12)}{\alpha_1} > \frac{\tilde{v}_2(2)}{\alpha_2} $ then $\pi  = N $ otherwise $\pi = \{\{1\},\{2\}\}$.
	
	\item $\alpha(1) =\alpha_1; \alpha(2) = \alpha_2; ~\alpha(12) = \alpha_2$. Probability of such alpha's is $p_1 p_2^2$. Thus, the actual values are $v_1(1) = \frac{\tilde{v}_1(1)}{\alpha_1}, ~~ v_2(2) = \frac{\tilde{v}_2(2)}{\alpha_2}; ~v_1(12) = \frac{\tilde{v}_1(12)}{\alpha_2}~ and ~ v_2(12) = \frac{\tilde{v}_2(12)}{\alpha_2}$. Since, $\frac{\tilde{v}_1(12)}{\alpha_2} > \tilde{v}_1(12) > \tilde{v}_1(1) > \frac{\tilde{v}_1(1)}{\alpha_1}$. The noise-free game is: ${v}_1(12) > {v}_1(1)$; \hspace{1 mm} ${v}_2(12) > {v}_2(2)$.
	So, $\pi = \tilde{\pi}$ in this case.
	
	\item $\alpha(1) =\alpha_2; \alpha(2) = \alpha_1; ~\alpha(12) = \alpha_1$. Probability of such alpha's is $p_1^2 p_2$. Thus, the actual values are $v_1(1) = \frac{\tilde{v}_1(1)}{\alpha_2}, ~~ v_2(2) = \frac{\tilde{v}_2(2)}{\alpha_1}; ~v_1(12) = \frac{\tilde{v}_1(12)}{\alpha_1}~ and ~ v_2(12) = \frac{\tilde{v}_2(12)}{\alpha_1}$ Clearly, the preferences in the noise-free game are not clear; it will depend on the relative values of $\alpha_1$, $\alpha_2$ and $\tilde{v}$. If $\alpha_1$, $\alpha_2$ and $\tilde{v}$'s are such that $\frac{\tilde{v}_1(12)}{\alpha_1} > \frac{\tilde{v}_1(1)}{\alpha_2} $ then $\pi  = N $ otherwise $\pi = \{\{1\},\{2\}\}$.

	\item $\alpha(1) =\alpha_2; \alpha(2) = \alpha_1; ~\alpha(12) = \alpha_2$. Probability of such alpha's is $p_1 p_2^2$. Thus, the actual values are $v_1(1) = \frac{\tilde{v}_1(1)}{\alpha_2}, ~~ v_2(2) = \frac{\tilde{v}_2(2)}{\alpha_1}; ~v_1(12) = \frac{\tilde{v}_1(12)}{\alpha_2}~ and ~ v_2(12) = \frac{\tilde{v}_2(12)}{\alpha_2}$. Since, $\frac{\tilde{v}_2(12)}{\alpha_2} > \tilde{v}_2(12) > \tilde{v}_2(2) > \frac{\tilde{v}_2(2)}{\alpha_1}$. The noise-free game is: ${v}_1(12) > {v}_1(1)$; \hspace{1 mm} ${v}_2(12) > {v}_2(2)$.
	So, $\pi = \tilde{\pi}$ in this case.
	
	\item $\alpha(1) =\alpha_2; \alpha(2) = \alpha_2; ~\alpha(12) = \alpha_1$. Probability of such alpha's is $p_1 p_2^2$. Thus, the actual values are $v_1(1) = \frac{\tilde{v}_1(1)}{\alpha_2}, ~~ v_2(2) = \frac{\tilde{v}_2(2)}{\alpha_2}; ~v_1(12) = \frac{\tilde{v}_1(12)}{\alpha_1}~ and ~ v_2(12) = \frac{\tilde{v}_2(12)}{\alpha_1}$. The noise-free game preferences are unclear; it will depend on the relative values $\alpha_1$, $\alpha_2$ and $\tilde{v}$. If $\alpha_1$, $\alpha_2$ and $\tilde{v}$'s are such that $\frac{\tilde{v}_1(12)}{\alpha_1} > \frac{\tilde{v}_1(1)}{\alpha_2} $ and $\frac{\tilde{v}_2(12)}{\alpha_1} > \frac{\tilde{v}_2(2)}{\alpha_2} $	then $\pi  = N $ otherwise $\pi = \{\{1\},\{2\}\}$. 

	\item $\alpha(1) =\alpha_2; \alpha(2) = \alpha_2; ~\alpha(12) = \alpha_2$. The probability of such alpha is $ p_2^3$. Thus, the actual values are $v_1(1) = \frac{\tilde{v}_1(1)}{\alpha_2}, ~~ v_2(2) = \frac{\tilde{v}_2(2)}{\alpha_2}; ~v_1(12) = \frac{\tilde{v}_1(12)}{\alpha_2}~ and ~ v_2(12) = \frac{\tilde{v}_2(12)}{\alpha_2}$. The noise-free game is: ${v}_1(12) > {v}_1(1)$; \hspace{1 mm} ${v}_2(12) > {v}_2(2)$.
	So, $\pi = \tilde{\pi}$ in this case.
\end{enumerate}
Since $\overline{r}= \max \left\lbrace \frac{\tilde{v}_1(12)}{\tilde{v}_1(1)}, \frac{\tilde{v}_2(12)}{\tilde{v}_2(2)} \right\rbrace$, and $\underline{r}=\min \left\lbrace \frac{\tilde{v}_1(12)}{\tilde{v}_1(1)}, \frac{\tilde{v}_2(12)}{\tilde{v}_2(2)} \right\rbrace$. From above cases, we see that in 14 out of 27 cases (case 1,3,4,7,8,9,12,13,15,17,21,23,25,27) we have $\pi= \tilde{\pi} = N$ in noise-free game. In these cases, the relative value of $\tilde{v}_1(\cdot), \tilde{v}_2(\cdot)$ should satisfy $\alpha_1 \geq \overline{r}, \frac{1}{\alpha_2} \geq \overline{r}, \frac{\alpha_1}{\alpha_2} \geq \overline{r}$. The prediction probability in this case is given below as $g(p_1, p_2)$. Whereas if we allow for the cases, say $\alpha_1 < \underline{r}~;~ \frac{1}{\alpha_2} < \underline{r} ~; ~ \frac{\alpha_1}{\alpha_2} < \underline{r}$, then the prediction probability is 1. So, these are the two extreme cases. However, if we take any other range of $\alpha$'s, the prediction probability will be more than $g(p_1, p_2)$ and less than 1. Thus,
\begin{equation}
\label{eqn: pred_prob_3_supp_suppli}
\mathbb{P}[\pi = \tilde{\pi} ~|~ game~1] = \begin{cases} g(p_1,p_2),  & ~if~~ \alpha_1 \geq \overline{r}~;~ \frac{1}{\alpha_2} \geq \overline{r} ~; ~ \frac{\alpha_1}{\alpha_2} \geq \overline{r}
\\
1, &~ if~~ \alpha_1 < \underline{r}~;~ \frac{1}{\alpha_2} < \underline{r} ~; ~ \frac{\alpha_1}{\alpha_2} < \underline{r},
\end{cases}
\end{equation}
where $g(p_1,p_2) =
p_1^3 + p_2^3 +   2(p_1 (1-p_1-p_2)^2 +  p_2^2 (1-p_1-p_2) + p_1p_2(1-p_1-p_2) + p_1p_2^2 ) + p_1^2 p_2 + p_1^2 (1-p_1-p_2) + p_2 (1-p_1-p_2)^2 +(1-p_1-p_2)^3$. \end{proof}

\subsubsection{Safety value via global minima for 2 agents and 3 support noise model}
\label{subsec: counter_example_3_support}

Here we will show that the above prediction probability given in Equation (\ref{eqn: pred_prob_3_supp_suppli}) can be non-convex in $p_1,p_2$. So, the global minima are difficult to hope for. 

Note that $\frac{\partial g(p_1,p_2) }{\partial p_1} = 3 p_1^2 - ( p_2 - 1)^2$ and $\frac{\partial g(p_1,p_2)}{\partial p_2} = -2p_1(p_2-1) - 3 p_2^2 + 6p_2 - 2$. Hence, we have $\frac{\partial^2 g(p_1,p_2) }{\partial^2 p_1} = 6p_1$, $\frac{\partial^2 g(p_1,p_2) }{\partial p_1 p_2} = \frac{\partial^2 g(p_1,p_2) }{\partial p_2 p_1} = -2 (p_2 -1 )$, and $\frac{\partial^2 g(p_1,p_2) }{\partial p_2^2} = -2p_1 - 6p_2 + 6$. Thus, the Hessian of $g(p_1, p_2)$ is
\begin{equation*}
H(g(p_1, p_2)) = 
\begin{bmatrix}
6 p_1  & -2 (p_2 -1 )
\\
-2 (p_2 -1 ) & -2p_1 - 6p_2 + 6
\end{bmatrix}.
\end{equation*}
For $p_1 = 0.3$ and $p_2 = 0.5$, we have
\begin{equation*}
H(g(p_1, p_2)) = 
\begin{bmatrix}
0.18  & 1
\\
1 & 2.4
\end{bmatrix}.
\end{equation*}
The eigenvalues are $\lambda_1 = 2.78$, and $\lambda_2 = -0.20$. So, $g(p_1,p_2)$ is not a convex function. Therefore, finding the global minima is difficult.

Though the above prediction probability is non-convex, one can get the noise set such that the prediction probability is more than a given satisfaction $\zeta$. Similar to the 2 support cases, where the prediction probability was a convex function, but the noise regimes were disjoint intervals, in 3 support cases also, we get disjoint sets. However, computing the exact safety value is problematic because it is the global minima of the non-convex prediction probability function. Note that the safety value is a fundamental limit such that below a user-given satisfaction $\zeta$, the partition is noise robust in the entire noise probability simplex. 

As earlier, in the noise regimes where the prediction probability is more than $\zeta$, a partition $\tilde{\pi}$ that is core-stable in a noisy game will remain core-stable in a noise-free game.

\end{document}